\documentclass{ieeeconf}
\IEEEoverridecommandlockouts
\usepackage{verbatim}
\usepackage{epsfig}
\usepackage{amsmath}
\usepackage{lipsum}
\usepackage{amsthm}
\usepackage{amssymb}
\usepackage{psfrag}
\usepackage{soul}
\usepackage[T1]{fontenc}
\usepackage[ruled]{algorithm2e}
\usepackage{mathtools, nccmath}
\usepackage{cuted}
\usepackage{multirow}
\usepackage[table,xcdraw]{xcolor}
\usepackage{etoolbox}
\usepackage{textcomp}
\usepackage{xcolor}
\usepackage{url}
\usepackage{dsfont}
\usepackage{mathtools}
\usepackage{bbm}

\usepackage[usenames,dviIBRames]{pstricks}
\usepackage{etoolbox}
\usepackage{graphicx}

\usepackage{stfloats}
\makeatletter
\patchcmd{\@makecaption}
  {\scshape}
  {}
  {}
  {}
\makeatother
\usepackage[noadjust]{cite}

\usepackage[font=footnotesize,skip=0pt]{caption}

\newtheorem{theorem}{Theorem}[]
\newtheorem{lemma}{Lemma}[]

\newcommand{\admm}{\textbf{\texttt{DC-DistADMM}}}
\newcommand{\ii}{i}
\newcommand{\jj}{j}
\newcommand{\s}{s}
\newcommand{\kk}{k}

\newcommand{\sst}{\mathrm{ss}}
\newcommand{\Vnom}{\overline{V}}
\newcommand{\Qnom}{\widehat{Q}}
\newcommand{\unom}{\overline{u}}

\newcolumntype{L}[1]{>{\raggedright\arraybackslash}p{#1}}

\DeclareMathOperator*{\argmin}{argmin}
\DeclareMathOperator*{\minimize}{minimize}

\definecolor{color1}{rgb}{0,0,0}

\begin{document}
\title{\textbf{Large-Signal Stability Analysis of Optimization-Based Secondary Control for Distributed Energy Resources}}

\author{Vivek Khatana$^{\dagger}$,
        Soham Chakraborty$^{\ddagger}$,
        Murti V. Salapaka$^{\wr}$
\thanks{$^{\dagger}$Department of Mechanical Science and Engineering, University of Illinois at Urbana-Champaign, IL, USA {\tt\scriptsize \{vkhatana\}@\{illinois.edu\}}, $^{\ddagger}$Department of Electrical Engineering, Indian Institute of Science, Karnataka, India {\tt\scriptsize \{schakraborty@iisc.ac.in\}}, $^{\wr}$Department of Electrical and Computer Engineering, University of Minnesota, MN, USA {\tt\scriptsize \{murtis@umn.edu\}}. The research conducted with the support of the United States Department of Energy via grant~DE-CR$0000040$.
}
}
\maketitle

\begin{abstract}
This article develops a large-signal stability analysis for a sampled-data optimization-based secondary controller for distributed energy resources (DERs) in power systems. The induced closed loop combines nonlinear inverter power-flow dynamics, filtered active and reactive power measurements, constrained optimization updates, and interpolation-based actuation between sampling instants. We study this optimization-in-the-loop nonlinear sampled-data system beyond local linearization. The analysis provides computable bounds on the voltage, filtered reactive power, and the secondary control input. We further characterize steady-state operating points and establish how the optimizer objectives and constraints connect voltage regulation with equal per-unitized reactive power sharing. Finally, input-to-state stability of the frequency dynamics is established with respect to DER voltages and control inputs. These results provide a rigorous mathematical foundation for sampled-data optimization-based secondary control of DERs.
\end{abstract}

\vspace{-0.2in}
\section{Introduction}
Distributed energy resources (DERs) operate in different functional modes depending on their interface and control objectives. Grid-following (GFL) DERs synchronize with the grid and inject prescribed power, whereas grid-forming (GFM) DERs establish and regulate voltage and frequency, particularly in islanded or weak-grid conditions~\cite{minai2024evolution}. Primary droop control enables decentralized operation of DERs~\cite{chandorkar1993control}; however, by itself, it generally does not restore frequency to its nominal value, does not guarantee accurate reactive power sharing under heterogeneous impedances, and does not explicitly enforce device-level operational constraints~\cite{zhong2011robust, savaghebi2012secondary,khatana_IECON_distcontrol}. These limitations motivate secondary control mechanisms~\cite{khatana_IECON_distcontrol,simpson2013synchronization,simpson2015secondary} that compute voltage adjustments for GFM-DERs and active-power references for GFL-DERs through constrained optimization problems.

This paper studies the closed-loop dynamics induced by a sampled-data optimization-based secondary controller for DER networks. The controller is designed to coordinate voltage regulation, reactive power sharing, and frequency support while respecting local operational constraints. Its broader architecture, distributed implementation, design rationale, plug-and-play operation, privacy-preserving information exchange, and performance validation are presented in a separate companion article. In contrast, the present paper is devoted to the mathematical analysis of the resulting closed-loop dynamics induced by the distributed secondary controller. The controller model and equations needed for the analysis are specified explicitly in Section~\ref{sec:control_setup}; consequently, the results developed here are self-contained and do not rely on external implementation details. The closed-loop system is mathematically challenging because the physical inverter dynamics evolve continuously, while the secondary controller is updated only at sampling instants. At each sampling time, constrained optimization problems are instantiated using measured voltage, reactive-power, and frequency signals. The computed optimizer outputs are then applied through interpolation over the next sampling interval. Hence, the resulting dynamics are not those of a standard continuous-time droop-controlled system, but of a nonlinear sampled-data system with optimization in the loop. Existing analyses of secondary control in power systems primarily focus on small-signal, local behavior and continuous-time models~\cite{simpson2013synchronization, simpson2015secondary}. As a result, they do not directly provide large-signal guarantees for the closed-loop voltage and power dynamics and are not applicable to the setting considered here.

We first formulate the closed-loop sampled-data model induced by the optimization-based secondary controller, capturing nonlinear power-flow coupling, filtered power measurements, constrained optimization updates, and interpolation-based actuation. We then provide a self-contained stability analysis of the resulting closed-loop model. The main results establish that i)~the DER voltage and filtered reactive-power dynamics remain within an explicit forward-invariant set, ii)~a positive steady state exists in this set, iii)~the steady state achieves voltage regulation and equal per-unitized reactive power sharing, and iv)~the frequency dynamics are input-to-state practically stable with respect to secondary control signals and DER voltages. Together, these results provide rigorous closed-loop guarantees for sampled-data optimization-based secondary control beyond local small-signal analysis. 

\vspace{-0.1in}
\section{Control Setup}\label{sec:control_setup}
We consider a network of DERs operating in either GFM or GFL mode. Both GFM- and GFL-DERs operate within a hierarchical control architecture~\cite{Hierarchical_Control} comprising three layers: zero-level, primary-level, and secondary-level control. The GFM-DERs provide voltage and frequency regulation through primary-level droop control, while the GFL-DERs inject active power through primary-level $\mathrm{PQ}$ dispatch~\cite{Hierarchical_Control}.

Although primary droop enables decentralized operation, it generally does not achieve the desired steady-state objectives~\cite{khatana_IECON_distcontrol}. In particular, $P$-$\omega$ droop induces a steady-state frequency deviation from the nominal value, while unequal line impedances prevent accurate reactive power sharing under $Q$-$V$ droop and cause GFM-DER terminal voltages to deviate from nominal values. Thus, a secondary-level controller is needed to restore frequency, regulate voltages, and enforce reactive power sharing while respecting device-level constraints~\cite{khatana_IECON_distcontrol}.

The secondary controller considered in this paper modifies the GFM-DER voltage references and the GFL-DER active-power references using sampled measurements and constrained distributed optimization. Next, we discuss the model of this secondary controller. The subsequent sections analyze the resulting nonlinear sampled-data closed-loop dynamics. 
References for $P$-$\omega$/$Q$-$V$ droop control law~\cite{mu_syn} for GFM-DER $\ii$ are given by
\begin{align}
    \omega_i &= \textstyle \overline{\omega} - r_{\omega_i} \overline{P}_i, \ \ \overline{P}_i =  \frac{1}{(\tau_{P_i} s+1)}P_i, \label{eq:w_droop}\\
     V_i &= \textstyle \Vnom - r_{V_i}\overline{Q}_i + U_i^\star,\ \ \overline{Q}_i =  \frac{1}{(\tau_{Q_i} s+1)}Q_i. \label{eq:v_droop_0}
\end{align}
Here, $\overline{\omega}$, $\omega_i, \Vnom, V_i$ are the nominal and measured frequency and voltages, respectively, and $r_{\omega_i}, r_{V_i} \in \mathbb{R}_{>0}$ are the droop coefficients. The instantaneous active and reactive power $P_i$ and $Q_i$, respectively, are given by 
\begin{align}
    \hspace{-0.1in} P_i &= \textstyle G_{ii} V_i^2 - \sum_{\kk \in N_i} B_{ik} V_i V_{k} (\delta_i - \delta_{k}) - P_i^\mathrm{gfl} \label{eq:activepower},\\
    \hspace{-0.12in}  Q_i & = \textstyle -B_{ii} V_i^2 +  \sum_{\kk \in N_i} B_{ik} V_i V_{k},\label{eq:reactivepower}
\end{align}
where, the complex admittance, $Y_{ik} = Y_{ki} \in \mathbb{C}$, between buses $i \in \{1,2,\dots,n\}$ and $\kk \in \{1,2,\dots,n\}$ are represented as $Y_{ik} := -j B_{i \kk} = -j B_{\kk i} =: Y_{ki}$ with $B_{i \kk} = B_{\kk i} < 0$. The set of neighbors of a bus $i \in \{1,2,\dots,n\}$ is defined as $N_i:= \{\kk | \kk \in \{1,2,\dots,n\}, \kk \neq i, Y_{i\kk} \neq 0 \}$ and $Y_{ii} = G_{ii} + j (B^{\mathrm{sh}}_{i} + \sum_{\kk \in N_i} B_{ik}) := G_{i i} + j B_{ii}$ where, $G_{i i} > 0$ is the shunt conductance, $B^{\mathrm{sh}}_{i} < 0$ is the shunt susceptance and $B_{i i} < 0$. The quantities $\overline{P}_i$ and $\overline{Q}_i$ in~\eqref{eq:w_droop}-\eqref{eq:v_droop_0} are the averaged active and reactive power injections of GFM-DER $i$, obtained by passing the instantaneous active and reactive power $P_i$ and $Q_i$ through low-pass filters with time constants, $\tau_{P_i}, \tau_{Q_i} \in \mathbb{R}_{>0}$ respectively. 

Signals, $U_i^\star(t), i \in \{1,2,\dots, n\}$ in~\eqref{eq:v_droop_0} are determined using secondary-level controller via the following design law, 
\begin{align}\label{eq:U_i_design_law}
   \hspace{-0.08in} U_i^\star(t) := u_i(t) - \beta_{Q_i} r_{V_i} \overline{Q}_i(t) + \beta_{V_i} (\Vnom - V_i(t)),
\end{align} 
where $\beta_{V_i}, \beta_{Q_i} \in \mathbb{R}$ are design hyper-parameters. With the compensation signals $U_i^\star$ in~\eqref{eq:U_i_design_law} the reference in~\eqref{eq:v_droop_0} becomes
\begin{align}\label{eq:voltage_reference}
    V_i &=  \textstyle \Vnom -  r_{V_i}\frac{(1 + \beta_{Q_i})}{(1 + \beta_{V_i})} \overline{Q}_i + \frac{1}{(1 + \beta_{V_i})}u_i,
\end{align}
where, $u_i(t), t\in(t_\s,t_{\s+1}]$, is updated from the last value $u_i(t_\s)$ via the first-order hold like interpolation as follows\footnote{Sampling instants satisfy $t_{\s+1}:=t_\s+\Delta_\s+\Delta$ for $\s\in\{0,1,2,\dots\}$}
\begin{equation*}
     u_i(t) = \hspace{-0.035in} \Bigg\{  \hspace{-0.02in} \begin{aligned} 
     & u_i(t_\s), t \in (t_\s, t_\s + \Delta_\s),\\ 
     & \textstyle u_i(t_\s) + \frac{(t - t_\s - \Delta_\s)(x_{\s,i} - u_{i}(t_\s))}{\Delta}, t \in (t_\s + \Delta_\s, t_{\s+1}]
    \end{aligned} 
\end{equation*}
for all $i\in\{1,\dots,n\}$.  Here $x_{\s,i}$ denotes the solution of the following problem instantiated with measurements $[V_i(t_\s), \ \overline{Q}_i(t_\s)] \in \mathbb{R}^2, i \in \{1,\dots, n\}$:
\begin{equation}\label{eq:gfm_optimization}
\begin{aligned}
    & \hspace{-0.65in} x_\s := \textstyle \argmin_{x_{\s,1}, \dots, x_{\s,n}}  \sum_{i=1}^n \ \frac{1}{2}(x_{\s,i} - \alpha_i (t_\s))^2 \\
    \mbox{subject to} \ x_{\s,i} &= x_{\s,j}, \ \forall i,\jj \in \{1,\dots,n\}, \\
     x_{\s,i} &\in D_i^{V_\Delta}(t_\s),\ \forall i \in \{1,\dots,n\} \\
     \alpha_i (t_\s) &:= (1+\beta_{Q_i})(\Vnom - V_i(t_\s)) \\ 
    & \hspace{0.15in} - \beta_{V_i} r_{V_i} \overline{Q}_i (t_\s) \ \forall i \in \{1,\dots,n\}.
\end{aligned}    
\end{equation}
Here, $D_i^{V_\Delta}(t_\s) := \big\{ \zeta \big | | \zeta - (1+\beta_{Q_i}) r_{V_i} \overline{Q}_i (t_\s) + \beta_{V_i} (\Vnom - V_i(t_\s)) | \leq V_\Delta \big\}$ with a prescribed parameter $V_\Delta \in \mathbb{R}_{>0}$. 

In addition,  the secondary-level controller designs active power references $P_i^\star$, $i \in \{1,\dots,n\}$, for the primary-level $\mathrm{PQ}$-dispatch controllers of the GFL-DERs as
\begin{align}\label{eq:Pi_star_design_law}
P_i^\star(t) := P_i^{\min} + p_i(t)(P_i^{\max} - P_i^{\min}).
\end{align}
Here, $P_i^{\min}$ and $P_i^{\max}$ denote the minimum and maximum active power generation capacities of GFL-DER $i$, respectively, and $p_i(t),t\in(t_\s,t_{\s+1}]$, are updated from $p_i(t_\s)$ as
\begin{equation*} 
    p_i(t) = \hspace{-0.035in} \Bigg\{  \hspace{-0.02in} \begin{aligned}
            & p_i(t_\s),  t \in (t_\s, t_\s + \Delta_\s),\\
          & \textstyle p_i(t_\s) + \frac{(t - t_\s - \Delta_\s)(y_{\s,i} - p_i(t_\s))}{\Delta}, t \in (t_\s + \Delta_\s, t_{\s+1}]
    \end{aligned}
\end{equation*}
for all $i\in\{1,\dots,n\}$, where $y_{\s,i}$ denotes the solution of
\begin{equation}\label{eq:gfl_optimization}
    \begin{aligned}
    y_\s := & \textstyle \argmin_{y_{\s,1},\dots, y_{\s,n}} \sum_{i=1}^{n} \ \frac{1}{2}(P_{y_{i}}^\star)^2 + b_i P_{y_{i}}^\star + c_i \\
   &\mbox{subject to} \  0 \leq y_{\s,i} \leq 1 \ \forall i \in \{1,\dots, n\}, \\
   & \hspace{0.6in} \textstyle \sum_{i=1}^{n} P_{y_{i}}^\star =  \rho_\mathrm{d}(t_\s) \\
   &\hspace{-0.35in} P_{y_{i}}^\star := P_i^{\min} + y_{\s,i} (P_i^{\max} - P_i^{\min}) \ \forall i \in \{1,\dots, n\},
\end{aligned}
\end{equation}
where $y_\s=[y_{\s,1},\dots,y_{\s,n}]\in\mathbb{R}^{n}$ denotes the solution of the optimization problem, and $P_i^{\min}$ and $P_i^{\max}$ are the active power limits of GFL-DER $i$, as in~\eqref{eq:Pi_star_design_law} and $\rho_\mathrm{d}(t_\s)$ is the total active power requirement. The GFL-DER power set-point estimate in~\eqref{eq:gfl_optimization} is parametrized as
$P_{y_i}^\star:=P_i^{\min}+y_{\s,i}(P_i^{\max}-P_i^{\min})$, with $0\leq y_{\s,i}\leq 1$, ensuring that the resulting active power set-points satisfy capacity constraints. The solution to~\eqref{eq:gfl_optimization}  determines the GFL-DER power reference $P_i^\star$ in~\eqref{eq:Pi_star_design_law}. The GFL-DER output $P_i^\mathrm{gfl}$ tracks $P_i^\star$ with negligible, assumed zero, error~\cite{yazdani} and therefore, 
\begin{align}\label{eq:gfl_power}
    \textstyle P_i^\mathrm{gfl} (t) = P_i^\star(t) = P_i^{\min} + p_i(t)(P_i^{\max} - P_i^{\min}) \ \forall t.
\end{align}
Collectively, problems~\eqref{eq:gfm_optimization} and~\eqref{eq:gfl_optimization} can be formulated as
\begin{equation}\label{eq:opt_prob_equivalent}
    \begin{aligned}
	& \textstyle \minimize_{\zeta_1, \zeta_2, \dots, \zeta_{n_\zeta}} \quad \sum_{i=1}^{n_\zeta} f_i(\zeta_i)\\
	& \mbox{subject to} \hspace{0.2in} \zeta_i = \zeta_\jj \ \forall i, \jj \in \{1,2,\dots,n_\zeta\}, \\
    & \hspace{0.15in} \zeta_i \in \mathcal{X}_i, \ \overline{H}_i \zeta_i  = \overline{h}_i, \ \underline{H}_i \zeta_i \leq \underline{h}_i  \  \forall i \in \{1,\dots,n_\zeta\}.
\end{aligned}
\end{equation}
Here, each DER maintains a local decision variable $\zeta_i$. The objective $f_i:\mathbb{R}^{n_\zeta}\rightarrow\mathbb{R}$, constraint set $\mathcal{X}_i\subseteq\mathbb{R}^{n_\zeta}$, equality constraints $\overline{H}_i\zeta_i=\overline{h}_i$, and inequality constraints $\underline{H}_i\zeta_i\leq\underline{h}_i$ are local to DER $i$ and can be enforced in a decentralized manner. The consensus constraints $\zeta_i=\zeta_\jj$ couple the local decisions and require network-level coordination. The developed secondary-level controller, therefore, solves~\eqref{eq:opt_prob_equivalent} using the distributed discrete-time $\admm$ algorithm developed in~\cite{ADMM_tac}. The $\admm$ algorithm has a geometric rate of convergence; after $\theta$ iterations,
\begin{align}\label{eq:admm_iter_comp}
\|\zeta_i^{(\theta)} - \zeta^\star \|^2 \leq \Upsilon (0.75)^{\theta}, \quad \mbox{for all} \ i,
\end{align}
where $\Upsilon$ is a known constant determined from problem data. Hence, an accuracy of $\epsilon$ requires only $\theta_\epsilon=O(\log(1/\epsilon))$ iterations, and therefore the sampling interval $\Delta_\s$ in the sampled-data secondary controller can be determined according to the chosen accuracy $\epsilon$. 
\vspace{-0.15in}
\section{Stability Analysis}\label{sec:stability_analysis}
Using~\eqref{eq:w_droop}-\eqref{eq:v_droop_0}, GFM-DER $i$ closed-loop dynamics are,
\begin{align}
    \dot{\delta}_i &= \omega_i, \quad \tau_{P_i}\dot{\omega_i} = -(\omega_i - \overline{\omega}) - r_{\omega_i}P_i, \label{eq:freq_dyan}\\
    \tau_{Q_i}\dot{V_i} &= -(V_i - \Vnom) - r_{V_i}Q_i + U_i^\star +  \tau_{Q_i}\dot{U}_i^\star. \label{eq:V_dyan}
\end{align}
\subsection{Analysis of GFM-DER Voltage Dynamics Loop}
\noindent Substituting~\eqref{eq:U_i_design_law} in~\eqref{eq:V_dyan} we get,
\begin{align}\label{eq:V_dot_expanded}
    \tau_{Q_i}\dot{V_i} &= -(V_i - \Vnom) - r_{V_i}Q_i + U_i^\star +  \tau_{Q_i}\dot{U}_i^\star \nonumber  \\
    & = -(V_i - \Vnom) - r_{V_i}Q_i + u_i - \beta_{Q_i} r_{V_i} \overline{Q}_i \\
    & \hspace{-0.15in} + \beta_{V_i} (\Vnom - V_i) + \tau_{Q_i}\dot{u}_i -  \tau_{Q_i}\beta_{Q_i} r_{V_i} \dot{Q}^\mathrm{avg}_i - \tau_{Q_i}\beta_{V_i} \dot{V}_i \nonumber. 
\end{align}
Let $\tilde{\beta}_{Q_i} := (1 + \beta_{Q_i}), \ \tilde{\beta}_{V_i} := (1 + \beta_{V_i}), \ \tilde{\beta}_i := \tilde{\beta}_{Q_i}/\tilde{\beta}_{V_i} \ \forall i.$
Substituting $\tau_{Q_i} \dot{\overline{Q}}_i = Q_i - \overline{Q}_i$ in~\eqref{eq:V_dot_expanded} we get for all $i$,
\begin{align}\label{eq:V_dot_expanded_intermediate}
    \tau_{Q_i}(1 + \beta_{V_i})\dot{V_i} &= -(1+\beta_{V_i})(V_i - \Vnom) - r_{V_i}Q_i \nonumber \\
    & \hspace{-0.4in} - \beta_{Q_i} r_{V_i} \overline{Q}_i -  \beta_{Q_i} r_{V_i} (Q_i - \overline{Q}_i) + u_i + \tau_{Q_i}\dot{u}_i \nonumber \\
    & \hspace{-0.8in} = -\tilde{\beta}_{V_i}(V_i - \Vnom) - \tilde{\beta}_{Q_i}r_{V_i}Q_i + u_i + \tau_{Q_i}\dot{u}_i,
\end{align}
Therefore, using~\eqref{eq:V_dot_expanded_intermediate} we have for all $i \in \{1,2,\dots,n\}$,
\begin{align}\label{eq:V_dot_expanded_final}
    \dot{V_i} = \textstyle -\frac{1}{\tau_{Q_i}}(V_i - \Vnom) - \frac{\tilde{\beta}_i r_{V_i}}{\tau_{Q_i}}Q_i + \frac{1}{\tau_{Q_i} \tilde{\beta}_{V_i}}u_i + \frac{1}{\tilde{\beta}_{V_i}}\dot{u}_i.
\end{align}
Let $V:=[V_1,\dots,V_n]^\top\in\mathbb{R}^n, \overline{Q}(V):=\overline{Q} =[\overline{Q}_1,\dots,$ $\overline{Q}_n]^\top\in\mathbb{R}^{n}$. We establish an invariance result for the closed-loop voltage dynamics around the nominal operating point $\mathbf{\Vnom}:=\Vnom\mathbf{1}_{n}$ and $\Qnom:=Q(\mathbf{\Vnom})$. The result provides an explicit forward-invariant set whose size is determined by the network parameters and controller design, characterizing closed-loop boundedness and performance.
\begin{theorem}\label{thm:centered_invariance}
Consider centered voltage and filtered reactive-power variables $\widetilde{V}(t) = V(t) - \mathbf \Vnom, \widetilde{q}(t) = \overline{Q}(t) - \Qnom$. Given radii $R_V>0$, $R_q>0$, define the set of centered states $\Omega(R_V,R_q) := \{(\widetilde V,\widetilde q):\|\widetilde V\|_2\le R_V,\;\|\widetilde q\|_2\le R_q\}$. Assume $\tau_{Q_i}\ge 1, \tilde\beta_{V_i}>0$ for all $i$. If $V_i>0$ for all $i$, then there exist constants $\alpha > 0, c_\star > 0, \nu > 0$ and $\rho_\Omega > 0$, defined explicitly from the system parameters and the radii $(R_V, R_q)$, satisfying 
$c_\star < \alpha \rho_\Omega$ and the sub-level set $\mathcal S_{\rho_\Omega}
:=
\{(\widetilde V,\widetilde q):\Psi(\widetilde V,\widetilde q)\le \rho_\Omega\}$ is forward invariant where, $\Psi(\widetilde V,\widetilde q)
:= \frac16 \widetilde V^\top \text{diag}(1/\tau_{Q_i}) \widetilde V
+\frac{\nu}{2}\|\widetilde q\|_2^2$. If $(\widetilde{V}(0), \widetilde{q}(0)) \in \Omega(R_V,R_q)$
then $\|\widetilde V(t)\|_2\le R_V, \mbox{and} \ \|\widetilde q(t)\|_2\le R_q \ \forall t \geq 0.$
\end{theorem}
\begin{proof}
Define $u(t) := [u_1(t), \dots, u_{n}(t)]^\top \in \mathbb{R}^{n}$, and $\dot{u}(t) := [\dot{u}_1(t), \dots, \dot{u}_{n}(t)]^\top \in \mathbb{R}^{n}$, $\unom_i :=\tilde\beta_{V_i}\tilde\beta_i r_{V_i} \Qnom_i$, for all $i$, $\unom := [\unom_1, \dots, \unom_{n}]^\top \in \mathbb{R}^{n}$, $d_i := |B_{i i}| + \textstyle \sum_{k\in N_i}|B_{i \kk}|$ for all $i$, and $d_B := \sqrt{\sum_{i=1}^{n} d_i^2}$. Fix $(R_V,R_q)>0$, define $\Omega(R_V,R_q) := \{(\widetilde V,\widetilde q):\|\widetilde V\|_2\le R_V,\;\|\widetilde q\|_2\le R_q\}, 
L_Q:= d_B \bigl(2\|\Vnom\|_2+R_V\bigr)$, and 
\begin{equation}\label{eq:constants}
\begin{aligned}
B_1 &= \max_i |1+\beta_{Q_i}|+\sqrt{n}\max_i |\beta_{V_i}|, \\
B_2 &= \max_i |\beta_{V_i}r_{V_i}|+\sqrt{n}\max_i |(1+\beta_{Q_i})r_{V_i}|, \\
B_0 &= \sqrt{n}V_\Delta + B_2\|\Qnom\|_2 + \|\unom\|_2, \\
X_\star &:=  \textstyle B_0+B_1R_V+B_2R_q, \  
D_\star:=\frac{2X_\star}{\Delta}. 
\end{aligned}
\end{equation}
Let $ \tau_{\min}:=\min_i \tau_{Q_i}, \tau_{\max}:=\max_i \tau_{Q_i}, \tilde\beta_{V,\min}:=\min_i \tilde \beta_{V_i},\boldsymbol{\tau}_{\beta r} := \text{diag}(\frac{\tilde{\beta}_i r_{V_i}}{\tau_{Q_i}}), \boldsymbol{\tau}_{\beta_V} := \text{diag}(1/(\tilde{\beta}_{V_i} \tau_{Q_i})), $ $ \boldsymbol{\beta}_V := \text{diag}(1/\tilde{\beta}_{V_i}), \boldsymbol{\tau} := \text{diag}(1/\tau_{Q_i}),$
and $K_u \le \frac{1}{3\tilde\beta_{V,\min}\tau_{\min}^2}, K_d\le \frac{1}{3\tilde\beta_{V,\min}\tau_{\min}}.$ Choose $\varepsilon_{r,1},\varepsilon_{r,2},\varepsilon_u,\varepsilon_d>0$, and 
$\varepsilon_q\in(0,2/\tau_{\max})$. Define $\varepsilon_r:=\varepsilon_{r,1}+\varepsilon_{r,2}$, and 
$d_r(R_V) := \textstyle  \frac{1}{18\varepsilon_{r,1}}
\|\boldsymbol{\tau}_{\beta r}\boldsymbol{\tau} \Qnom \|_2^2  + 
\frac{L_Q^2}{18\varepsilon_{r,2}}
\|\boldsymbol{\tau}\boldsymbol{\tau}_{\beta r}\mathbf{\Vnom}\|_2^2,  + \frac{1}{3} |
\mathbf{\Vnom}^\top
\boldsymbol{\tau}\boldsymbol{\tau}_{\beta r}\Qnom|, 
c_V := \textstyle \frac{1}{3\tau_{\max}^2}-\frac{\varepsilon_r+\varepsilon_u+\varepsilon_d}{2}$.
$\varepsilon_{r,1},\varepsilon_{r,2},\varepsilon_u,\varepsilon_d>0$ be chosen such that $c_V > 0$. Further, let $\nu>0, \varepsilon_q > 0$ be such that
$0<\nu< \frac{2\varepsilon_q\tau_{\min}^2\,c_V}{L_Q^2}.$ Let $a_V := \textstyle c_V-\nu\frac{L_Q^2}{2\varepsilon_q\tau_{\min}^2}, \ a_q := \nu\Bigl(\frac{1}{\tau_{\max}}-\frac{\varepsilon_q}{2}\Bigr), 
M_\Psi := \textstyle \max\left\{\frac{1}{6\tau_{\min}},\frac{\nu}{2}\right\}, \ 
m_\Psi:=\min\left\{\frac{1}{6\tau_{\max}},\frac{\nu}{2}\right\}, 
\alpha := \textstyle \frac{\min\{a_V,a_q\}}{M_\Psi}, \
c_\star := \textstyle d_r(R_V)
+\frac{K_u^2}{2\varepsilon_u}X_\star^2
+\frac{K_d^2}{2\varepsilon_d}D_\star^2,
\rho_\Omega := \textstyle \min\left\{\frac{R_V^2}{6\tau_{\max}},
\frac{\nu R_q^2}{2}
\right\}.$
We divide the proof into several steps.

\textbf{Step 1: Centered dynamics.}
Since $ V = \widetilde{V} + \mathbf{\Vnom}, 
\overline{Q}=\widetilde q + \Qnom,$
the voltage dynamics~\eqref{eq:V_dot_expanded} can be written as
\begin{align}
\hspace{-0.1in} \dot{\widetilde V}
=
-\boldsymbol\tau \widetilde V
-\boldsymbol\tau_{\beta r}\bigl(Q(V) - \Qnom\bigr) + \boldsymbol\tau_{\beta_V}(u - \unom)
+\boldsymbol\beta_V \dot u. \label{eq:centered_voltage_dyn_proof}
\end{align}
Also, using $\tau_{Q_i}\dot{\overline{Q}}_i = -\overline{Q}_i + Q_i,$
\begin{equation}\label{eq:centered_filter_dyn_proof}
\dot{\widetilde q}
= -\boldsymbol{\tau} \widetilde q + \boldsymbol{\tau}\bigl(Q(V) - \Qnom \bigr).
\end{equation}

\textbf{Step 2: Bound the optimizer $x_\s$.}
Because of consensus constraints $x_{\s,i} = x_{\s,\jj}$, every feasible solution has the form $ x_\s = c_{x_\s} \mathbf{1}_{n}$,
for some scalar $c_{x_\s} \in \mathbb R$. Therefore, $\|x\|_2 = \sqrt{n} |c_{x_\s}|.$
Next, from~\eqref{eq:gfm_optimization},
$ \alpha_i (t_\s)
=
(1+\beta_{Q_i})\Vnom - (1+\beta_{Q_i})V_i - \beta_{V_i}r_{V_i}Q_i^{\mathrm{avg}}(t_\s).$
Thus,
\begin{align*}
\|\alpha \|_2
\le &  \max_i |1+\beta_{Q_i}| \|\mathbf{\Vnom} - V\|_2 +  
\max_i |\beta_{V_i}r_{V_i}| \|\overline{Q}\|_2.
\end{align*}
Let $L_\alpha^V := \max_i |1+\beta_{Q_i}|, 
L_\alpha^Q := \textstyle \max_i |\beta_{V_i}r_{V_i}|.$
Then, $\|\alpha\|_2
\le L_\alpha^V \|\mathbf{\Vnom} - V\|_2 + L_\alpha^Q \|\overline{Q}\|_2.$
Define, $ \overline{\alpha} := \frac{1}{n}\sum_{i=1}^{n} \alpha_i$,
and by Cauchy-Schwarz inequality,
\begin{align*}
|\bar \alpha|
= \textstyle
\frac{1}{n} |\mathbf{1}_{n}^\top \alpha|
& \textstyle \le
\frac{1}{\sqrt{n}}\|\alpha\|_2 \leq \textstyle \frac{L_\alpha^V}{\sqrt n}\|\mathbf{\Vnom} - V\|_2
+
\frac{L_\alpha^Q}{\sqrt n}\|\overline{Q}\|_2.
\end{align*}
Now consider sets $D_i^{V_\Delta}$ in~\eqref{eq:gfm_optimization}, with the centers given by $m_i =
(1+\beta_{Q_i})r_{V_i}Q_i^{\mathrm{avg}}(t_\s)
- \beta_{V_i}(\Vnom - V_i (t_\s))$
Therefore, $|m_i| \le |\beta_{V_i}||\Vnom -  V_i (t_\s)| + |(1+\beta_{Q_i})r_{V_i}||Q_i^{\mathrm{avg}}(t_\s)|.$
Taking the maximum over $i$ yields $ \|m\|_\infty
\le m_V\|V(t_\s) - \Vnom\|_2 + m_Q\|\overline{Q}(t_\s)\|_2$,
where $m_V := \max_i |\beta_{V_i}|, m_Q := \textstyle \max_i |(1+\beta_{Q_i})r_{V_i}|.$
Therefore, the solution $x_\s = c_{x_\s} \mathbf{1}_{n}$ is bounded, $ |c_{x_\s}|
\le
|\bar\alpha| + \|m\|_\infty + V_\Delta.$
Substituting the previously derived bounds gives $|c_{x_\s}|
\leq \left(\frac{L_\alpha^V}{\sqrt{n}} + m_V \right) \|V(t_\s) - \mathbf{\Vnom}\|_2
+ V_\Delta + \textstyle  \left(\frac{L_\alpha^Q}{\sqrt{n}} + m_Q \right) \|\overline{Q}(t_\s)\|_2$
Multiplying by $\sqrt{n}$ and using $\|x_\s \|_2=\sqrt{n} |c_{x_\s}|$, we conclude
\begin{align}\label{eq:bound_on_xs}
\|x_\s \|_2
\le B_1\|\tilde{V}(t_\s)\|_2 + B_2\|\overline{Q}(t_\s)\|_2 + \sqrt{n}V_\Delta,
\end{align}
where $B_1,B_2$ are exactly those given in~\eqref{eq:constants}. 
Hence, by the triangle inequality,
\begin{align}
\|x_\s - \unom\|_2
&\le
\|x_\s\|_2+\|\unom\|_2 \le
\sqrt{n}V_\Delta
+
B_1\|\widetilde V(t_\s)\|_2
\nonumber \\ & \hspace{0.2in} +
B_2\|\widetilde q(t_\s)\|_2 +
B_2\|\Qnom\|_2
+
\|\unom\|_2 \nonumber\\
&=
B_0 + B_1\|\widetilde V(t_\s)\|_2+B_2\|\widetilde q(t_\s)\|_2.
\label{eq:x_minus_unom_bound_proof}
\end{align}
Now assume $(\widetilde V(t_\s),\widetilde q(t_\s))\in \mathcal{S}_{\rho_\Omega}$, then $ \|\widetilde V(t_\s)\|_2\le R_V, \ \|\widetilde q(t_\s)\|_2\le R_q,$
and so~\eqref{eq:x_minus_unom_bound_proof} yields
\begin{equation}\label{eq:x_minus_unom_bound_region_proof}
\|x_\s - \unom\|_2
\le
B_0+B_1R_V+B_2R_q
=
X_\star.
\end{equation}

\textbf{Step 3: Explicit bounds on $u - \unom$ and $\dot{u}$.}
For $t\in[t_\s,t_\s + \Delta_\s)$, the interpolation law gives $ u(t) = u(t_\s)$. And for $t\in[t_\s + \Delta_\s,t_{\s+1})$, the interpolation law gives
\begin{align*}
u(t)
= \textstyle \Bigl(1-\frac{t - t_\s - \Delta_\s}{\Delta}\Bigr)u(t_\s)
+ \frac{t - t_\s - \Delta_\s }{\Delta}x_\s.
\end{align*}
Subtracting $\unom$ for $t \in [t_\s, t_{\s+1})$, we have
$u(t) - \unom$
$= \begin{cases}
u(t_\s) - \unom \\
\big(1-\frac{t-t_\s-\Delta_\s}{\Delta}\big)(u(t_\s) - \unom)
+
\frac{t-t_\s-\Delta_\s}{\Delta}(x_s - \unom).
\end{cases} 
$ 

Since both coefficients are nonnegative and sum to one,
\begin{align}
\|u(t) - \unom\|_2
&\le \textstyle 
\Bigl(1-\frac{t-t_\s-\Delta_\s}{\Delta}\Bigr)\|u(t_s)- \unom\|_2 \nonumber 
\\
& \textstyle \hspace{0.2in}+ \frac{t-t_\s-\Delta_\s}{\Delta}\|x_\s - \unom\|_2 \nonumber\\
& \hspace{-0.4in}\le
\max \bigl\{\|u(t_\s) - \unom\|_2, \|x_\s - \unom\|_2 \bigr\}.\label{eq:u_convex_bound_proof}
\end{align}
Assume $\|u(0) - \unom\|_2\le X_\star$. Then, using~\eqref{eq:x_minus_unom_bound_region_proof}, we prove by induction over the sampling intervals $t \in [t_\s, t_{\s+1})$ that
\begin{equation}\label{eq:u_minus_unom_global_bound_proof}
\|u(t) - \unom\|_2\le X_\star.
\end{equation}
as long as $(\widetilde V(t),\widetilde q(t))\in \mathcal{S}_{\rho_\Omega}$. Indeed, it is true at $t=0$ by assumption. If it holds at $t=t_\s$, then~\eqref{eq:u_convex_bound_proof} and~\eqref{eq:x_minus_unom_bound_region_proof} imply it holds for all $t\in[t_\s,t_{\s+1})$. Also, $\dot u(t)= \begin{cases}
    0 & t\in[t_\s + \Delta_\s)\\
    \frac{x_\s - u(t_\s)}{\Delta} & t\in[t_\s + \Delta_\s,t_{\s+1})
\end{cases}$.

\begin{align}
\hspace{-0.6in}\mbox{Hence,} \ \|\dot u(t)\|_2
& \textstyle \le \frac{1}{\Delta}\bigl(\|x_\s - \unom\|_2+\|u(t_\s) -\unom\|_2\bigr) \nonumber\\
& \textstyle \le \frac{1}{\Delta}(X_\star+X_\star)
= \frac{2X_\star}{\Delta}
= D_\star. \label{eq:udot_bound_proof}
\end{align}
Therefore, on $\Omega(R_V,R_q)$,
\begin{equation}\label{eq:u_udot_explicit_bounds_proof}
\|u(t) - \unom \|_2\le X_\star,
\quad
\|\dot u(t)\|_2 \le D_\star.
\end{equation}

\textbf{Step 4: Local Lipschitz bound on $Q(V) - \Qnom$.}
From the quadratic structure of the reactive-power map,
$\|Q(V) - Q(W)\|_2
\le
d_B\bigl(\|V\|_2+\|W\|_2\bigr)\|V-W\|_2.$
Apply this with $W = \mathbf{\Vnom}$. Then, $\|Q(V) - \Qnom\|_2
\le
d_B\bigl(\|V\|_2+\|\mathbf \Vnom\|_2\bigr)\|V - \Vnom\|_2.$
Since $V=\widetilde V+\mathbf \Vnom, \ 
\|V\|_2\le \|\widetilde V\|_2+\|\mathbf \Vnom\|_2,
$
and on $\Omega(R_V,R_q)$ we have $\|\widetilde V\|_2\le R_V$, it follows that $\|V\|_2\le R_V+\|\mathbf \Vnom\|_2.$
Hence
\begin{align}
\|Q(V) - \Qnom\|_2
&\le
d_B\bigl(R_V+\|\mathbf \Vnom\|_2+\|\mathbf \Vnom\|_2\bigr)\|\widetilde V\|_2 \nonumber\\
&\hspace{-0.5in} =
d_B\bigl(2\|\mathbf \Vnom\|_2+R_V\bigr)\|\widetilde V\|_2 = L_Q\|\widetilde V\|_2.
\label{eq:LQ_bound_proof}
\end{align}

\textbf{Step 5: Derivative of the voltage Lyapunov function.}
Define $\Phi(\widetilde V):=\frac16 \widetilde V^\top \boldsymbol\tau \widetilde V.$ Then $\dot\Phi
=
\dot{\widetilde{V}}^\top \frac16 \boldsymbol\tau \widetilde V
+
\widetilde V^\top \frac16 \boldsymbol\tau \dot{\widetilde V}
= \frac13 \widetilde V^\top \boldsymbol\tau \dot{\widetilde V}$,
since $\boldsymbol\tau$ is diagonal and symmetric. Using~\eqref{eq:centered_voltage_dyn_proof},
\begin{align}
\dot\Phi
&= \textstyle -\frac13 \widetilde V^\top \boldsymbol\tau^2 \widetilde V
-\frac13 \widetilde V^\top \boldsymbol\tau\boldsymbol\tau_{\beta r}(Q - \Qnom) \nonumber\\
&\hspace{0.3in} \textstyle 
+\frac13 \widetilde V^\top \boldsymbol\tau\boldsymbol\tau_{\beta_V}(u - \unom)
+\frac13 \widetilde V^\top \boldsymbol\tau\boldsymbol\beta_V\dot u.
\label{eq:Phi_dot_centered_intermediate_proof}
\end{align}
To make the reactive term symmetric, write
\begin{align}
& \hspace{-0.5in} \dot\Phi
= \textstyle
-\frac13 \widetilde V^\top \boldsymbol\tau^2 \widetilde V
+\mathcal{T}_Q
+\frac13 \widetilde V^\top \boldsymbol\tau\boldsymbol\tau_{\beta_V}(u - \unom) +\frac13 \widetilde V^\top \boldsymbol\tau\boldsymbol\beta_V \dot u,
\label{eq:Phi_dot_centered_with_TQ_proof} \\
\hspace{-0.1in} \mbox{where} \ \mathcal{T}_Q
 & \textstyle := -\frac13 \widetilde V^\top \boldsymbol\tau\boldsymbol\tau_{\beta r}(Q - \Qnom) \nonumber \\
 & \hspace{-0.25in} \textstyle = 
-\frac16 \widetilde V^\top \boldsymbol\tau\boldsymbol\tau_{\beta r}(Q-\Qnom)
-\frac16 (Q - \Qnom)^\top \boldsymbol\tau_{\beta r}\boldsymbol\tau \widetilde V. \label{eq:TQ_def_proof}
\end{align}

We now estimate $\mathcal{T}_Q$ using the improved centered decomposition. Since $\widetilde V = V - \mathbf \Vnom$,
we expand~\eqref{eq:TQ_def_proof} as
\begin{align}
&\textstyle 
-\frac16 (V - \mathbf \Vnom)^\top \boldsymbol\tau\boldsymbol\tau_{\beta r}(Q - \Qnom) -\frac16 (Q - \Qnom)^\top \boldsymbol\tau_{\beta r}\boldsymbol\tau (V - \mathbf \Vnom) \nonumber\\
& = \textstyle 
-\frac16 V^\top \boldsymbol\tau\boldsymbol\tau_{\beta r}Q
-\frac16 Q^\top \boldsymbol\tau_{\beta r}\boldsymbol\tau V 
+\frac16 V^\top \boldsymbol\tau\boldsymbol\tau_{\beta r} \Qnom
 \nonumber\\
&\textstyle \hspace{0.2in} +\frac16 (\Qnom)^\top \boldsymbol\tau_{\beta r}\boldsymbol\tau V 
+\frac16 (\mathbf \Vnom)^\top \boldsymbol\tau\boldsymbol\tau_{\beta r}Q
+\frac16 Q^\top \boldsymbol\tau_{\beta r}\boldsymbol\tau \mathbf \Vnom \nonumber\\
&\hspace{0.2in} \textstyle 
-\frac16 (\mathbf \Vnom)^\top \boldsymbol\tau\boldsymbol\tau_{\beta r}\Qnom
-\frac16 (\Qnom)^\top \boldsymbol\tau_{\beta r}\boldsymbol\tau \mathbf \Vnom. \label{eq:TQ_expanded_full_proof}
\end{align}
Rearranging terms, we have
\begin{align}
\mathcal{T}_Q &= \textstyle -\frac16 V^\top \boldsymbol\tau\boldsymbol\tau_{\beta r}Q
-\frac16 Q^\top \boldsymbol\tau_{\beta r}\boldsymbol\tau V +\frac13 (\Qnom)^\top \boldsymbol\tau_{\beta r}\boldsymbol\tau \widetilde V \nonumber\\
& \textstyle \hspace{0.2in}
+\frac13 (\mathbf \Vnom)^\top \boldsymbol\tau\boldsymbol\tau_{\beta r}Q.
\label{eq:TQ_rearranged_proof}
\end{align}
Using $Q(V)=\mathrm{diag}(V)\mathcal B V$,
and defining $\Sigma(V):=\mathrm{diag}\left(\tilde\beta_i r_{V_i}V_i/2\tau_{Q_i}^2\right)$,
the first two terms in~\eqref{eq:TQ_rearranged_proof} become
$-\frac13 V^\top \bigl(\Sigma(V)\mathcal B+\mathcal B^\top \Sigma(V)\bigr)V.$
Note that the matrix $\Sigma(V) \mathcal{B} + \mathcal{B}^\top \Sigma(V)$ is positive semi-definite (see Lemma~\ref{lem:psd_proof}),
and therefore $
\textstyle -\frac13 V^\top \bigl(\Sigma(V)\mathcal B+\mathcal B^\top \Sigma(V)\bigr)V \le 0.$
Hence,
\begin{align}
\mathcal{T}_Q
& \textstyle \le
\frac13 (\Qnom)^\top \boldsymbol\tau_{\beta r}\boldsymbol\tau \widetilde V
+\frac13 (\mathbf \Vnom)^\top \boldsymbol\tau\boldsymbol\tau_{\beta r}Q. \label{eq:TQ_after_PSD_drop_proof}
\end{align}
We now bound the two remaining terms. For the first term, Young's inequality gives
\begin{align}
\left| \textstyle 
\frac13 (\Qnom)^\top \boldsymbol\tau_{\beta r}\boldsymbol\tau \widetilde V
\right|
& \textstyle \le
\frac{\varepsilon_{r,1}}{2}\|\widetilde V\|_2^2
+
\frac{1}{18\varepsilon_{r,1}}
\|\boldsymbol\tau_{\beta r}\boldsymbol\tau \Qnom\|_2^2.
\label{eq:TQ_term1_bound_proof}
\end{align}
For the second term, using~\eqref{eq:LQ_bound_proof},
\begin{align}
& \left| \textstyle
\frac13 (\mathbf \Vnom)^\top \boldsymbol\tau\boldsymbol\tau_{\beta r}Q
\right| \le \textstyle
\frac13
\|\boldsymbol\tau\boldsymbol\tau_{\beta r}\mathbf \Vnom\|_2
\|Q - \Qnom\|_2 + \left| \textstyle
\frac{1}{3} \mathbf \Vnom^\top \boldsymbol\tau\boldsymbol\tau_{\beta r}\Qnom
\right| \nonumber \\
& \textstyle \le
\frac13
\|\boldsymbol\tau\boldsymbol\tau_{\beta r}\mathbf \Vnom\|_2
L_Q\|\widetilde V\|_2 + \left| \textstyle
\frac{1}{3} \mathbf \Vnom^\top \boldsymbol\tau\boldsymbol\tau_{\beta r}\Qnom
\right| \nonumber\\
& \textstyle \le
\frac{\varepsilon_{r,2}}{2}\|\widetilde V\|_2^2
+
\frac{L_Q^2}{18\varepsilon_{r,2}}
\|\boldsymbol\tau\boldsymbol\tau_{\beta r}\mathbf \Vnom\|_2^2 + \left| \textstyle
\frac{1}{3} \mathbf \Vnom^\top \boldsymbol\tau\boldsymbol\tau_{\beta r}\Qnom
\right|.
\label{eq:TQ_term2_bound_proof}
\end{align}
Combining~\eqref{eq:TQ_after_PSD_drop_proof},~\eqref{eq:TQ_term1_bound_proof}, and~\eqref{eq:TQ_term2_bound_proof}, we get
\begin{align}
\mathcal{T}_Q
& \le \textstyle \frac{\varepsilon_r}{2}\|\widetilde V\|_2^2+d_r(R_V),
\label{eq:TQ_final_bound_proof}
\end{align}
where $\varepsilon_r=\varepsilon_{r,1}+\varepsilon_{r,2}$. Next, the quadratic term satisfies
\begin{align}
\textstyle -\frac13 \widetilde V^\top \boldsymbol\tau^2 \widetilde V
& \textstyle \le
-\frac{1}{3\tau_{\max}^2}\|\widetilde V\|_2^2,
\label{eq:tau2_dissipation_proof}
\end{align}
because the smallest eigenvalue of $\boldsymbol\tau^2$ is $1/\tau_{\max}^2$. For the control term involving $u - \unom$, using~\eqref{eq:u_udot_explicit_bounds_proof} and Young's inequality,
\begin{align}
\textstyle \frac13 \widetilde V^\top \boldsymbol\tau\boldsymbol\tau_{\beta_V}(u - \unom)
&\le
K_u\|\widetilde V\|_2\|u - \unom\|_2 \le \textstyle
\frac{\varepsilon_u}{2}\|\widetilde V\|_2^2 \nonumber
\\ 
& \hspace{-0.4in} \textstyle +
\frac{K_u^2}{2\varepsilon_u}\|u - \unom\|_2^2 \le
\frac{\varepsilon_u}{2}\|\widetilde V\|_2^2
+
\frac{K_u^2}{2\varepsilon_u}X_\star^2.
\label{eq:u_term_bound_proof}
\end{align}
Similarly, for the $\dot u$ term,
\begin{align}
\textstyle \frac13 \widetilde V^\top \boldsymbol\tau\boldsymbol\beta_V\dot u
& \textstyle \le
K_d\|\widetilde V\|_2\|\dot u\|_2 \le
\frac{\varepsilon_d}{2}\|\widetilde V\|_2^2
+
\frac{K_d^2}{2\varepsilon_d}\|\dot u\|_2^2 \nonumber\\
&\textstyle \le
\frac{\varepsilon_d}{2}\|\widetilde V\|_2^2
+
\frac{K_d^2}{2\varepsilon_d}D_\star^2.
\label{eq:udot_term_bound_proof}
\end{align}
Substituting~\eqref{eq:TQ_final_bound_proof}, \eqref{eq:tau2_dissipation_proof}, \eqref{eq:u_term_bound_proof}, and \eqref{eq:udot_term_bound_proof} into~\eqref{eq:Phi_dot_centered_with_TQ_proof} yields
\begin{align}
\dot\Phi
&\textstyle \le
-\left(
\frac{1}{3\tau_{\max}^2}
-\frac{\varepsilon_r+\varepsilon_u+\varepsilon_d}{2}
\right)\|\widetilde V\|_2^2
+d_r(R_V) \nonumber\\
&\hspace{0.2in}
\textstyle +
\frac{K_u^2}{2\varepsilon_u}X_\star^2
+
\frac{K_d^2}{2\varepsilon_d}D_\star^2 \nonumber\\
&=
\textstyle -c_V\|\widetilde V\|_2^2
+d_r(R_V)
+
\frac{K_u^2}{2\varepsilon_u}X_\star^2
+
\frac{K_d^2}{2\varepsilon_d}D_\star^2.
\label{eq:Phi_dot_final_proof}
\end{align}

\textbf{Step 6: Derivative of the filter-energy term.}
From~\eqref{eq:centered_filter_dyn_proof}, we have $\frac12\frac{d}{dt}\|\widetilde q\|_2^2
= \widetilde q^\top \dot{\widetilde q} =
-\widetilde q^\top \boldsymbol{\tau}\widetilde q
+ \widetilde q^\top \boldsymbol{\tau}(Q-\Qnom).$
Since
$ \textstyle \boldsymbol{\tau} \succeq \frac{1}{\tau_{\max}}I,
\
\|\boldsymbol{\tau}\|\le \frac{1}{\tau_{\min}},
$
we obtain $\frac12\frac{d}{dt}\|\widetilde q\|_2^2 \le
-\frac{1}{\tau_{\max}}\|\widetilde q\|_2^2
+
\frac{1}{\tau_{\min}}\|\widetilde q\|_2\|Q - \Qnom\|_2.$
Using~\eqref{eq:LQ_bound_proof}, $\frac12\frac{d}{dt}\|\widetilde q\|_2^2 \le -\frac{1}{\tau_{\max}}\|\widetilde q\|_2^2
+ \frac{L_Q}{\tau_{\min}}\|\widetilde q\|_2\|\widetilde V\|_2$.
Applying Young's inequality with parameter $\varepsilon_q$, $\frac{L_Q}{\tau_{\min}}\|\widetilde q\|_2\|\widetilde V\|_2
\le \frac{\varepsilon_q}{2}\|\widetilde q\|_2^2
+ \frac{L_Q^2}{2\varepsilon_q\tau_{\min}^2}\|\widetilde V\|_2^2.$
\begin{align}
\hspace{-0.1in} \mbox{Thus,} \ \textstyle \frac12\frac{d}{dt}\|\widetilde q\|_2^2
& \textstyle \le
-\Bigl(\frac{1}{\tau_{\max}}-\frac{\varepsilon_q}{2}\Bigr)\|\widetilde q\|_2^2
+
\frac{L_Q^2}{2\varepsilon_q\tau_{\min}^2}\|\widetilde V\|_2^2.
\label{eq:q_energy_final_proof}
\end{align}

\textbf{Step 7: Derivative of the composite Lyapunov function and invariance.}
Define $\Psi(\widetilde V,\widetilde q)
=
\Phi(\widetilde V)+\frac{\nu}{2}\|\widetilde q\|_2^2.$
Multiplying~\eqref{eq:q_energy_final_proof} by $\nu$ and adding to~\eqref{eq:Phi_dot_final_proof}, we get $\dot\Psi
\le
-\left(
c_V-\nu\frac{L_Q^2}{2\varepsilon_q\tau_{\min}^2}
\right)\|\widetilde V\|_2^2 -\nu\Bigl(\frac{1}{\tau_{\max}}-\frac{\varepsilon_q}{2}\Bigr)\|\widetilde q\|_2^2
+ c_\star,$
where
$c_\star= d_r(R_V)
+\frac{K_u^2}{2\varepsilon_u}X_\star^2 +\frac{K_d^2}{2\varepsilon_d}D_\star^2.$
Due to the choice of $\nu$, $a_V=
c_V-\nu\frac{L_Q^2}{2\varepsilon_q\tau_{\min}^2}>0.$
Also, since $\varepsilon_q<2/\tau_{\max}$, $a_q=
\nu\Bigl(\frac{1}{\tau_{\max}}-\frac{\varepsilon_q}{2}\Bigr)>0.$
Therefore,
\begin{align}
\dot\Psi
&\le
-a_V\|\widetilde V\|_2^2-a_q\|\widetilde q\|_2^2+c_\star.
\label{eq:Psi_dot_with_aVaQ_proof}
\end{align}
Next, by definition of $M_\Psi$, $\Psi(\widetilde V,\widetilde q)
\le M_\Psi\bigl(\|\widetilde V\|_2^2+\|\widetilde q\|_2^2\bigr).$
Hence, $a_V\|\widetilde V\|_2^2+a_q\|\widetilde q\|_2^2
\ge \frac{\min\{a_V,a_q\}}{M_\Psi}\Psi
= \alpha \Psi.$
Substituting this into~\eqref{eq:Psi_dot_with_aVaQ_proof}, we obtain
\begin{align}
\dot\Psi
&\le
-\alpha\Psi+c_\star.
\label{eq:Psi_dot_scalar_proof}
\end{align}
Since, $c_\star < \alpha \rho_\Omega$, 
whenever $\Psi=\rho_\Omega$, $ \dot\Psi\le -\alpha\rho_\Omega+c_\star<0.$ 
The estimates leading to
\eqref{eq:Psi_dot_scalar_proof} are valid whenever
$(\widetilde V(t),\widetilde q(t))\in \Omega(R_V,R_q)$,
and for almost all $t$, since $u(t)$ is piecewise affine. Define the
first exit time
$T^\star
:=
\inf\left\{
t\ge 0:
(\widetilde V(t),\widetilde q(t))\notin \Omega(R_V,R_q)
\right\}.$
If no such time exists, then $T^\star=\infty$. On the interval $[0,T^\star)$, the trajectory belongs to
$\Omega(R_V,R_q)$. Therefore, all bounds derived above apply, and
\eqref{eq:Psi_dot_scalar_proof} holds for almost all
$t\in[0,T^\star)$: $\dot\Psi(t)\le -\alpha\Psi(t)+c_\star.$ By the comparison lemma, for all $t\in[0,T^\star), \Psi(t)
\le
e^{-\alpha t}\Psi(0)
+
\frac{c_\star}{\alpha}\bigl(1-e^{-\alpha t}\bigr).$
If $\Psi(0)\le \rho_\Omega$, then using $c_\star<\alpha\rho_\Omega.$
Then, for every $t\in[0,T^\star)$,
\[
\Psi(t)
\le
e^{-\alpha t}\rho_\Omega
+
\frac{c_\star}{\alpha}\bigl(1-e^{-\alpha t}\bigr)
\le
\rho_\Omega .
\]
Thus, $\Psi(t)\le \rho_\Omega, 
\forall t\in[0,T^\star)$.
Next, by definition of $\Psi$,
$\Psi(\widetilde V,\widetilde q)
= \textstyle 
\frac{1}{6} \widetilde V^\top \boldsymbol\tau \widetilde V
+
\frac{\nu}{2}\|\widetilde q\|_2^2.$
Finally, by definition of $m_\Psi$,
$\textstyle \Psi(\widetilde V,\widetilde q)
\ge
\frac{1}{6\tau_{\max}}\|\widetilde V\|_2^2, 
\Psi(\widetilde V,\widetilde q)
\ge
\frac{\nu}{2}\|\widetilde q\|_2^2.
$
Therefore, for all $t\in[0,T^\star)$, $\|\widetilde V(t)\|_2^2
\le
6\tau_{\max}\Psi(t)
\le
6\tau_{\max}\rho_\Omega
\le
R_V^2,$
and
$\|\widetilde q(t)\|_2^2
\le
\frac{2}{\nu}\Psi(t)
\le
\frac{2}{\nu}\rho_\Omega
\le
R_q^2.$
Hence,
$(\widetilde V(t),\widetilde q(t))\in \Omega(R_V,R_q),
\
\forall t\in[0,T^\star).$
Suppose, for contradiction, that $T^\star<\infty$. Since the state trajectory is continuous, taking the limit $t \to T^\star$ gives
$\|\widetilde V(T^\star)\|_2\le R_V, \ 
\|\widetilde q(T^\star)\|_2\le R_q.$
Thus
\[
(\widetilde V(T^\star),\widetilde q(T^\star))
\in \Omega(R_V,R_q),
\]
which contradicts the definition of $T^\star$ as the first exit time
from $\Omega(R_V,R_q)$. Therefore, $T^\star=\infty$. Consequently,
$\Psi(t)\le \rho_\Omega,
\ \forall t\ge 0$,
and hence
$ \|\widetilde V(t)\|_2\le R_V,
\
\|\widetilde q(t)\|_2\le R_q,
\
\forall t\ge 0.
$
Therefore, $\Omega(R_V,R_q)$ is forward invariant for all trajectories
starting in the sublevel set
$\mathcal S_{\rho_\Omega}
:=
\{(\widetilde V,\widetilde q):\Psi(\widetilde V,\widetilde q)\le \rho_\Omega\}.
$ Moreover,
$\mathcal S_{\rho_\Omega}\subseteq \Omega(R_V,R_q)$,
and $\mathcal S_{\rho_\Omega}$ is forward invariant.
This completes the proof.
\end{proof}

We next establish the existence of a steady state within the invariant region. Since the invariance confines trajectories to a compact set, fixed-point arguments can be used to characterize equilibrium behavior. We show that the closed-loop voltage dynamics admit a steady-state operating point consistent with the network power flow and controller structure, and that this equilibrium lies within the invariant region.

\begin{theorem}\label{thm:steady_state_existence}
Assume the hypotheses of Theorem~\ref{thm:centered_invariance} hold, and let $R_V<\Vnom$ be the radius therein. Let $\mathcal{K}_V
:=
\prod_{i=1}^{n}[\Vnom-R_V, \ \Vnom+R_V]
\subset \mathbb R^{n}$. Define the secondary-level voltage control map~\eqref{eq:voltage_reference} component-wise, for all $i=1,\dots,n$, as $T_i(V)
:= \textstyle 
\Vnom
+ 
\frac{1}{(1+\beta_{V_i})}x_{\s,i}(V,Q(V)) - \frac{r_{V_i}(1 + \beta_{Q_i})}{(1+\beta_{V_i})}Q_i(V).$ 
Then,
i) $T$ is continuous on $\mathcal K_V$, and $T(\mathcal K_V)\subseteq \mathcal K_V$, ii) there exists $V^\star\in\mathcal K_V$ such that $T(V^\star)=V^\star$. Let $Q^\star := Q(V^\star),
\ (\overline{Q})^\star:=Q^\star,$ and $u^\star:= x_\s(V^\star,Q^\star)$,
then signals $V(t)\equiv V^\star,\
Q(t)\equiv Q^\star,\
\overline{Q}(t)\equiv Q^\star,\
u(t)\equiv u^\star $ 
form a positive steady state of the voltage/filter/controller subsystem. 
\end{theorem}
\begin{proof}
Define $d_{\max} := \max_i \textstyle  \left(|B_{ii}|+\sum_{k\in N_i}|B_{i\kk}|\right), M_Q := d_{\max}(\Vnom $ $ + R_V)^2, M_u := B_1\sqrt{n} R_{V} + B_2\sqrt{n}M_Q + \sqrt{n}V_\Delta$, where, $B_1,B_2$ are as in~\eqref{eq:constants}
$\tilde\beta_{Q,\max}:=\max_i |\tilde\beta_{Q_i}|, r_{V,\max}:=\max_i |r_{V_i}|$. Let
\begin{equation}\label{eq:selfmap_condition}
\textstyle \frac{M_u+\tilde\beta_{Q,\max}r_{V,\max}M_Q}{\tilde\beta_{V,\min}}
\le R_V.
\end{equation}

We proceed in several steps.

\textbf{Step 1: $\mathcal K_V$ is a nonempty compact convex positive set.}
Since $R_{V}>0$, the set $\mathcal K_V$ is nonempty. Because it is a product of closed bounded intervals, it is compact; it is convex. Since $\Vnom > R_{V} >0$. 
Hence, every $V\in\mathcal K_V$ is componentwise positive:
$V_i\in [\Vnom-R_{V},\Vnom+R_{V}]
\subset (0,\infty),
\ i=1,\dots,n.$

\textbf{Step 2: Map $T$ is continuous.}
By definition~\eqref{eq:reactivepower}, $Q(V)$ is continuous for any $V$ and the solution $x_\s$ of~\eqref{eq:gfm_optimization} is continuous in $(V,Q(V))$ (see Lemma~\ref{lem:continuous_sol}). Therefore, the composition $V \mapsto x_\s\bigl(V,Q(V)\bigr)$ 
is continuous. Since each component $T_i$ is obtained from continuous operations, $T$ is continuous.

\textbf{Step 3: Uniform bounds for $Q(V)$ on $\mathcal K_V$.}
Let $V\in \mathcal K_V$. Since each component satisfies
$|V_i|\le \Vnom+R_{V}$, we have
$\|V\|_\infty\le \Vnom+R_{V}$. 
For each $i$, using the reactive-power relation~\eqref{eq:reactivepower}, we obtain $|Q_i(V)| \le    \left(|B_{ii}|+\sum_{\kk\in N_i}|B_{i\kk}|\right)\|V\|_\infty^2 \le d_{\max}(\Vnom+R_{V})^2 = M_Q.$
\begin{equation}\label{eq:Qi_bound_proof}
\hspace{-1in} \mbox{Hence,} \hspace{0.2in} |Q_i(V)|\le M_Q,
\quad i=1,\dots,n.
\end{equation}
which yields $\|Q(V)\|_2
\le
\sqrt{n}M_Q.$
Using~\eqref{eq:bound_on_xs}, we have $ \|x_\s \|_2
\le B_1\|\tilde{V}\|_2 + B_2\|\overline{Q}\|_2 + \sqrt{n}V_\Delta \le \left(B_1\sqrt{n} R_{V} + B_2\sqrt{n}M_Q + \sqrt{n}V_\Delta \right):= M_u.$


\textbf{Step 4: $T(\mathcal K_V)\subseteq \mathcal K_V$.}
Fix $V\in\mathcal K_V$. For each $i$, using~\eqref{eq:selfmap_condition}, and~\eqref{eq:Qi_bound_proof}, we get
\begin{align*}
|T_i(V)-\Vnom|
&= \textstyle \left|
\frac{x_{\s,i}(V,Q(V))-\tilde\beta_{Q_i}r_{V_i}Q_i(V)}{\tilde\beta_{V_i}}
\right| \le \\
& \hspace{-0.72in} \textstyle \frac{
|x_{\s,i}(V,Q(V))|
+
|\tilde\beta_{Q_i}|
|r_{V_i}|
|Q_i(V)|
}{\tilde\beta_{V_i}} \le
\frac{M_u+\tilde\beta_{Q,\max}r_{V,\max}M_Q}{\tilde\beta_{V,\min}}
\le
R_{V}
\end{align*}
where the last inequality is exactly~\eqref{eq:selfmap_condition}. Therefore
$\Vnom-R_{V}\le T_i(V)\le \Vnom+R_{V},
\ i=1,\dots,n,
$
which shows that $T(V)\in \mathcal K_V$. Since $V\in\mathcal K_V$ is arbitrary, we conclude
$T(\mathcal K_V)\subseteq \mathcal K_V.$

\textbf{Step 5: Existence of a fixed point.}
The set $\mathcal K_V$ is nonempty, compact, and convex, and since
$T:\mathcal K_V\to\mathcal K_V$ is continuous, Brouwer's fixed-point theorem~\cite{brouwer1911abbildung} implies that there exists $V^\star\in\mathcal K_V$ such that $ T(V^\star)=V^\star$.

\textbf{Step 6: Construction of the steady state.}
Define $Q^\star:=Q(V^\star), \ 
(\overline{Q})^\star:=Q^\star, \ 
u^\star:=x_\s(V^\star,Q^\star)$.
We claim that the constant signals
$V(t) \equiv V^\star, \
Q(t)\equiv Q^\star, \ \overline{Q}(t) \equiv Q^\star, \ u(t)\equiv u^\star,
$
solve the voltage/filter/controller subsystem. First, since $u(t) = u^\star = x_\s(V^\star, Q^\star)  = c^\star \mathbf 1_{\mathrm{gfm}}$, for some constant $c^\star$, for all $t$, we have
$\dot u_i(t)=0, i=1,\dots,n.$ Second, since $(\overline{Q})^\star=Q^\star$, the averaging filter
$\tau_{Q_i}\dot{(\overline{Q})}^\star = Q^\star_i - (\overline{Q})^\star$ 
yields
$ \dot{(\overline{Q})}_i^\star(t)=0, i=1,\dots,n.$ It remains to be verified that the voltage equation holds. The fixed-point identity $T(V^\star)=V^\star$ means that, for each $i$,
$V_i^\star
= \Vnom +
\frac{u_i^\star - \tilde\beta_{Q_i}r_{V_i}Q_i^\star}{\tilde\beta_{V_i}}.$
Equivalently, $u_i^\star
=
\tilde\beta_{V_i}(V_i^\star - \Vnom)+\tilde\beta_{Q_i}r_{V_i}Q_i^\star.$
Using the relation $\tilde\beta_i=\tilde\beta_{Q_i}/\tilde\beta_{V_i}$, we may rewrite this as
$0
= -\frac{1}{\tau_{Q_i}}(V_i^\star-\Vnom)
-\frac{\tilde\beta_i r_{V_i}}{\tau_{Q_i}}Q_i^\star
+\frac{1}{\tau_{Q_i}\tilde\beta_{V_i}}u_i^\star.$
Since $\dot u_i=0$, it implies $\dot{V}_i = 0, i=1,\dots,n$ for all $i$. Therefore, the voltage equation~\eqref{eq:V_dot_expanded_final} achieves a steady-state.

\textbf{Step 7: Positivity.}
Since $V^\star\in \mathcal K_V$, $V_i^\star\ge \Vnom-R_{V}>0, i=1,\dots,n.$ Thus, the steady state voltage is positive.

Hence, the constructed constant signals form a positive steady state of the voltage/filter/controller subsystem. 
\end{proof}

Having established steady-state existence, we characterize its relation to the secondary-level objectives. We show that the equilibrium induced by the proposed distributed controller coordinates GFM-DERs to achieve voltage regulation and equal per-unitized reactive power sharing.

\begin{theorem}\label{thm:steady_state}
Voltage dynamics~\eqref{eq:V_dyan} under the secondary-level control~\eqref{eq:U_i_design_law} achieve equal reactive power sharing and voltage regulation among GFM-DERs at steady state.
\end{theorem}
\begin{proof}
By Theorem~\ref{thm:steady_state_existence}, the closed-loop voltage dynamics admit a steady-state operating point. At this equilibrium, the measurements entering~\eqref{eq:gfm_optimization} are constant, so $\alpha_i(t_\s)$ and $D_i^{V_\Delta}(t_\s)$ are fixed at all steady-state sampling instants. Hence, problem~\eqref{eq:gfm_optimization} has the same optimizer, denoted $x^\sst$, at every such instant. By the sampled-data update law, after one update beyond steady state, $u_i(t_\s)=x_i^\sst$ for all sufficiently large $t_\s$ and all $i\in{1,\dots,n}$. Consequently, $x_i^\sst-u_i(t_\s)=0$, and the interpolation law gives $\dot u_i(t)=0$ for all sufficiently large $t$. Since $U_i^\star$ differs from $u_i$ only by steady-state constant terms, $\dot U_i^\star=0$ for all $i\in{1,\dots,n}$ at steady state. Setting $\dot{V}_i = 0$ in~\eqref{eq:V_dyan} and using  $\dot U_i^\star = 0, \overline{Q}_i^{\sst} = Q_i^{\sst}$ gives
\begin{align}\label{eq:Vi_steady_state}
\hspace{-0.08in} (V_i^{\sst} - \Vnom) =  u_i^{\sst} + \beta_{V_i} (\Vnom - V_i^{\sst}) - \tilde{\beta}_{Q_i} r_{V_i}Q_i^{\sst}.
\end{align}
Let $x^\sst$ be $\admm$ solution with accuracy $\varepsilon/4$. Then, $|V_i^{\sst} - \Vnom| = \big|u_i^{\sst} + \beta_{V_i} (\Vnom - V_i^{\sst})- \tilde{\beta}_{Q_i} r_{V_i}Q_i^{\sst}\big| \leq |x^\sst + \beta_{V_i} (\Vnom - V_i^{\sst}) - \tilde{\beta}_{Q_i} r_{V_i}Q_i^{\sst}| + |u_i^{\sst} - x^\sst| \leq V_\Delta + \textstyle \sqrt{\Upsilon (0.75)^{\theta_{\varepsilon}}} \leq V_\Delta + \varepsilon/2$,
where, we used~\eqref{eq:admm_iter_comp}. From~\eqref{eq:Vi_steady_state},
\begin{align}\label{eq:steady_state}
(1 + \beta_{V_i})(V_i^{\sst} - \Vnom) + (1 + \beta_{Q_i}) r_{V_i}Q_i^{\sst} = u_i^{\sst}, \ \forall i. 
\end{align}
Since the steady-state inputs $u_i^\sst$ are obtained from the $\admm$ solution $x_i^\sst$~\cite{ADMM_tac}, the agreement error satisfies
$|u_i^{\sst} - u_\jj^{\sst}| \leq \textstyle \sqrt{\Upsilon (0.75)^{\theta_{\varepsilon}}} \leq \varepsilon$ for all $i,\jj$. Thus, with suitable design parameters, the secondary-level inputs~\eqref{eq:U_i_design_law} achieve the objectives of tasks $\mathcal{T}_1$ and $\mathcal{T}_2$ at steady state: 
\begin{itemize}
    \item [(i)] Let $\beta_{V_i} = \beta_{V} = - 1 + \mu$, $0< \mu \ll 0.01$, $\beta_{Q_i} = \beta_{Q} = 0$, for all GFM-DERs. Then from~\eqref{eq:steady_state},  $|r_{V_i} Q_i^{\sst} - r_{V_\jj} Q_\jj^{\sst}| \approx |u_i^{\sst} - u_\jj^{\sst}| \leq \varepsilon \ \forall i,\jj$. Hence, equal per-unitized reactive power sharing is achieved among all GFM-DERs. 
    \item [(ii)] Let $\beta_{V_i} = \beta_{V} = 0, \beta_{Q_i} = \beta_{Q} = - 1 + \mu$, $0< \mu \ll 0.01$, for all GFM-DERs $i$. Then from~\eqref{eq:steady_state}, $|V_i^{\sst}  - V_\jj^{\sst}| = | (\Vnom + u_i^{\sst}) - (\Vnom + u_\jj^{\sst}) | \approx |u_i^{\sst} - u_\jj^{\sst}| \leq \varepsilon$ for all GFM-DERs $i,\jj$. Thus, the GFM-DERs buses have similar voltage magnitudes in steady state. Moreover, if $\beta_{V_i}=\beta_V=0$ and $\beta_{Q_i}=\beta_Q=-1+\mu$, with $0<\mu\ll0.01$, for all GFM-DERs $i$, then the solution $x_i^\sst$ of~\eqref{eq:gfm_optimization}, and hence $u_i^\sst$, satisfies $u_i^\sst\approx0$. Therefore, the steady-state voltage magnitude satisfies $V_i=V_{\mathrm{nom}}+u_i^\sst\approx\Vnom$.
\end{itemize}
Therefore, with appropriate design choices, the secondary-level controller achieves equal reactive power sharing
and voltage regulation among GFM-DERs at steady state.
\end{proof}
\subsection{Analysis of GFM-DER Frequency Control Loop}
Next, we establish the stability properties of frequency dynamics~\eqref{eq:freq_dyan}. Let $\delta=[\delta_1,\dots,\delta_{n}]^\top\in\mathbb{R}^{n}, \omega=[\omega_1,\dots,\omega_{n}]^\top\in\mathbb{R}^{n}$ be the system states. Dynamics~\eqref{eq:freq_dyan},  using~\eqref{eq:activepower} and~\eqref{eq:gfl_power}, can be written compactly as
\begin{align}\label{eq:phase_omega_closedloop_vector}
    \dot{\begin{bmatrix} \delta \\ \omega \end{bmatrix}} &= \mathbf{A} \begin{bmatrix} \delta \\ \omega \end{bmatrix} + \mathbf{B} \begin{bmatrix} 0 \\ \Omega \end{bmatrix}  + \mathbf{C} \begin{bmatrix} 0 \\ \boldsymbol{p} \end{bmatrix} + \mathbf{D} \begin{bmatrix} 0 \\ \widehat{V} \end{bmatrix},
\end{align}
where, $\Omega = [\Omega_1, \dots, \Omega_{n}]^\top \in \mathbb{R}^{n}$ with $\Omega_i :=   (\overline{\omega} + r_{\omega_i} P_i^{\min})/\tau_{P_i}$, $\boldsymbol{p} = [p_1, \dots, p_{n}]^\top \in \mathbb{R}^{n}$, $\widehat{V} = [V^2_1, \dots, V^2_{n}]^\top \in \mathbb{R}^{n}$ and matrices $\mathbf{A},\mathbf{B}, \mathbf{C}, \mathbf{D}$ are
\begin{align}
    \mathbf{A} &:= \begin{bmatrix}
     \hspace{-0.05in} 0_{n}  \ \hspace{0.1in}  \mathbb{I}_{n} \\
     \mathbf{G}_{\omega} \ \hspace{0.05in} \boldsymbol{\tau}_P 
    \end{bmatrix}
    \mathbf{B}:= \begin{bmatrix}
     0_{n} \ \ 0_{n} \\
     0_{n}  \ \ \mathbb{I}_{n} 
    \end{bmatrix}
    \mathbf{C} := \begin{bmatrix}
     0_{n} \ \ 0_{n} \\
     0_{n} \  \ \boldsymbol{P}
    \end{bmatrix} \nonumber \\
    \mathbf{D} &:= \begin{bmatrix}
     0_{n} & 0_{n} \\
     0_{n} & \text{diag}((-G_{i i} r_{\omega_i})/\tau_{P_i})
    \end{bmatrix}, \label{eq:ABC_matrix_omega_delta}
\end{align}  
where, $\boldsymbol{\tau}_{P} := \text{diag}\left(-1/\tau_{P_i}\right), [\mathbf{G}_\omega]_{i \jj} := -\frac{r_{\omega_i}B_{i \jj} V_{i} V_{\jj}}{\tau_{P_i}}, $ $[\mathbf{G}_\omega]_{ii} := -\sum_{\jj} [\mathbf{G}_\omega]_{i \jj}, \boldsymbol{P} := \text{diag}(r_{\omega_i}(P_i^{\max} - P_i^{\min})/\tau_{P_i})$. Note that matrix $\mathbf{G}_\omega$ has a marginal mode associated with the vector $\mathbf{1}_{n}$. Since the phase angles are invariant under uniform shifts, we study the stability of~\eqref{eq:phase_omega_closedloop_vector} in relative coordinates. Let $\xi\in\mathbb{R}^{n}_{>0}$ denote a normalized left null vector of $\mathbf{G}_\omega$,  i.e., $\xi^\top \mathbf{G}_\omega=0$ and $\xi^\top\mathbf{1}_{n}=1$, and define the projection matrix
$\Pi:=\mathbb{I}_{n} -\mathbf{1}_{n}\xi^\top$. The projected phase and frequency variables are given by $ \tilde{\delta}:=\Pi\delta, \ 
\tilde{\omega}:=\Pi\omega. $
Because $\mathbf{G}_\omega\mathbf{1}_{n}=0$, we have $\mathbf{G}_\omega\delta=\mathbf{G}_\omega\tilde{\delta}$. Applying the projection to~\eqref{eq:phase_omega_closedloop_vector} yields
\[
\begin{aligned}
\dot{\tilde{\delta}} &= \tilde{\omega},\\
\dot{\tilde{\omega}}
&= \mathbf{G}_\omega\tilde{\delta}
+ \Pi\boldsymbol{\tau}_P \tilde{\omega}
+ \Pi\boldsymbol{\tau}_P\mathbf{1}_{n}\omega_\xi
+ \Pi\Omega
+ \Pi\boldsymbol{P}\boldsymbol{p}\\
& \hspace{0.15in} + \Pi\operatorname{diag}\left(-G_{ii}r_{\omega_i}/\tau_{P_i}\right)\widehat{V},
\end{aligned}
\]
where $\omega_\xi:=\xi^\top\omega$ denotes the weighted average frequency component. In particular, if the active-power filter time constants are identical, i.e., $\tau_{P_i}=\tau_P$ for all $i$, then $\boldsymbol{\tau}_P=-(1/\tau_P)\mathbb{I}_{n}$ and the term $\Pi\boldsymbol{\tau}_P\mathbf{1}_{n}\omega_\xi$ vanishes. In this case, the projected dynamics close in $(\tilde{\delta},\tilde{\omega})$ as
\begin{equation} \label{eq:centered_freq_delta}
\begin{aligned}
\hspace{-0.11in}\dot{\tilde{\delta}} &= \tilde{\omega}, \\
\hspace{-0.11in}\dot{\tilde{\omega}}
&= \textstyle \mathbf{G}_\omega\tilde{\delta}
-\frac{1}{\tau_P}\tilde{\omega}
+ \Pi\left[\Omega
+ \boldsymbol{P}\boldsymbol{p}+\operatorname{diag}\left(\frac{-G_{ii}r_{\omega_i}}{\tau_{P_i}}\right)\widehat{V}\right].
\end{aligned}
\end{equation}

Thus, the marginal absolute-angle direction is removed, and the stability analysis can be carried out on the disagreement subspace. The projected dynamics~\eqref{eq:centered_freq_delta}, with $\tilde{x}:= [\tilde{\delta}\ \tilde{\omega}]^\top$, can be expressed in a compact form as
\begin{equation}\label{eq:projected_compact}
  \begin{aligned}
    \dot{\tilde{x}} &= \mathbf{A}_\Pi \tilde{x} + \mathbf{B}_\Pi d(t), \\
    \mathbf{A}_\Pi &:=
\begin{bmatrix}
0_{n} & \mathbb{I}_{n} \\
\mathbf{G}_\omega & -\frac{1}{\tau_P}\mathbb{I}_{n}
\end{bmatrix},
\
\mathbf{B}_\Pi
:=
\begin{bmatrix}
0_{n}\\
\Pi
\end{bmatrix}, \\
d(t)
&:= \textstyle 
\Omega
+
\boldsymbol{P}\boldsymbol{p}(t)
+
\operatorname{diag}\left(\frac{-G_{ii}r_{\omega_i}}{\tau_P}\right)\widehat{V}(t).
\end{aligned}  
\end{equation}
We have the following result.

\begin{theorem}\label{thm:isps_projected_frequency}
Consider dynamics in~\eqref{eq:projected_compact}. Define the disagreement subspace as $\mathcal{S}
:=
\left\{
\tilde{x}\in\mathbb{R}^{2n}
|
\xi^\top\tilde{\delta}=0,
\xi^\top\tilde{\omega}=0
\right\},
$
where $\xi$ is a normalized left null vector of $\mathbf{G}_\omega$. Suppose the conditions of Lemma~\ref{lem:A_pi_hurwitz} hold, so that $\mathbf{A}_\Pi$ restricted to $\mathcal{S}$ is Hurwitz. The projected frequency dynamics are input-to-state stable~\cite{jiang1994small} on $\mathcal{S}$ with respect to the effective input $\Pi d(t)$. In particular, there exist constants $\kappa \geq 1$, $\lambda>0$, and $\gamma>0$ such that
$
\|\tilde{x}(t)\|
\leq
\kappa e^{-\lambda t}\|\tilde{x}(0)\|
+
\gamma
\sup_{0\leq s\leq t}\| \Pi d(s)\|,
\ t\geq 0.
$
Equivalently, if projected input satisfies $\sup_{t\geq 0}\|\Pi d(t)\|\leq \Delta_d$,
then
$
\|\tilde{x}(t)\|
\leq
\kappa e^{-\lambda t}\|\tilde{x}(0)\|
+ \gamma \Delta_d, \ t\geq 0.
$
Therefore, the projected phase-frequency dynamics remain ultimately bounded, with the ultimate bound proportional to the size of the projected control input.
\end{theorem}

\begin{proof}

First note that $\mathcal{S}$ is invariant under~\eqref{eq:projected_compact}. Indeed, if $\tilde{x}\in\mathcal{S}$, then $\xi^\top\dot{\tilde{\delta}}
=
\xi^\top\tilde{\omega}=0,
$
and using $\xi^\top\mathbf{G}_\omega=0$ and $\xi^\top\Pi=0$,
$\xi^\top\dot{\tilde{\omega}}
=
\xi^\top\mathbf{G}_\omega\tilde{\delta}
-
\frac{1}{\tau_P}\xi^\top\tilde{\omega}
+
\xi^\top\Pi d(t)
=0.$
Hence trajectories initialized in $\mathcal{S}$ remain in $\mathcal{S}$. Note that under the assumptions in Lemma~\ref{lem:A_pi_hurwitz},  $\mathbf{A}_\Pi$ is Hurwitz, and thus for any symmetric positive definite matrix $\mathbf{Q}$ on $\mathcal{S}$, there exists a symmetric positive definite matrix $\mathbf{M}$ on $\mathcal{S}$ such that $\mathbf{A}_\Pi^\top\mathbf{M}+\mathbf{M}\mathbf{A}_\Pi=-\mathbf{Q}$ on $\mathcal{S}$. Consider the Lyapunov function
$
W(\tilde{x})=\tilde{x}^\top \mathbf{M}\tilde{x}.
$
Along trajectories of the projected system,
$
\dot{W}
= \tilde{x}^\top
\left(
\mathbf{A}_\Pi^\top\mathbf{M}
+
\mathbf{M}\mathbf{A}_\Pi
\right)
\tilde{x}
+
2\tilde{x}^\top \mathbf{M}\mathbf{B}_\Pi d(t) =
-\tilde{x}^\top\mathbf{Q}\tilde{x}
+
2\tilde{x}^\top \mathbf{M}\mathbf{B}_\Pi d(t).
$
Applying Young's inequality, with $0 < \eta < 1,$ we get $ 2\tilde{x}^\top \mathbf{M}\mathbf{B}_\Pi d(t)
\leq \eta\lambda_{\min}(\mathbf{Q})\|\tilde{x}\|^2
+
\frac{\|\mathbf{M}\|^2}{\eta\lambda_{\min}(\mathbf{Q})}
\|\Pi d(t)\|^2.$
Thus, $\dot{W}
\leq
-(1-\eta)\lambda_{\min}(\mathbf{Q})\|\tilde{x}\|^2
+
\frac{\|\mathbf{M}\|^2}{\eta\lambda_{\min}(\mathbf{Q})}
\|\Pi d(t)\|^2.$
Since, for all $\tilde{x} \in \mathcal{S}$ we have, 
$\lambda_{\min}(\mathbf{M})\|\tilde{x}\|^2
\leq
W(\tilde{x})
\leq
\lambda_{\max}(\mathbf{M})\|\tilde{x}\|^2,$
there exist constants $c_1>0$ and $c_2>0$ such that
$\dot{W}
\leq
-c_1 W
+
c_2\|\Pi d(t)\|^2.$
Applying the comparison lemma gives $W(t)
\leq
e^{-c_1 t}W(0)
+
\frac{c_2}{c_1}
\sup_{0\leq s\leq t}\|\Pi d(s)\|^2.$
Using the quadratic bounds on $W$ yields
$\|\tilde{x}(t)\| \leq \textstyle 
 \sqrt{\frac{\lambda_{\max}(\mathbf{M})}{\lambda_{\min}(\mathbf{M})}}
e^{\frac{-c_1}{2}t}\|\tilde{x}(0)\| +
\sqrt{\frac{c_2}{c_1\lambda_{\min}(\mathbf{M})}}
\sup_{0\leq s\leq t}\|\Pi d(s)\|$. Thus, $\kappa := \sqrt{\frac{\lambda_{\max}(\mathbf{M})}{\lambda_{\min}(\mathbf{M})}}, \lambda = \frac{c_1}{c_2}, \gamma := \sqrt{\frac{c_2}{c_1\lambda_{\min}(\mathbf{M})}}$. Hence, projected dynamics are input-to-state stable with respect to the projected input $\Pi d$.
\end{proof}

\begin{theorem}\label{thm:isps_full_frequency}
Consider the phase-frequency dynamics whose projected form is given by~\eqref{eq:projected_compact}. Suppose the conditions of Lemma~\ref{lem:A_pi_hurwitz} hold and $\mathcal{S}$ is the subspace as defined in Theorem~\ref{thm:isps_projected_frequency}. Define the weighted average-frequency error $e_\xi(t):=\xi^\top\omega(t)-\omega_{\rm nom}$. Then $e_\xi(t)$ satisfies $\dot{e}_\xi(t)
= -\frac{1}{\tau_P}e_\xi(t)
+
d_\xi(t),$
where $d_\xi(t) := \xi^\top d(t)-\frac{1}{\tau_P}\omega_{\rm nom}$.
Hence, $ |e_\xi(t)|
\leq
e^{-t/\tau_P}|e_\xi(0)|
+
\tau_P
\sup_{0\leq s\leq t}|d_\xi(s)|$.
Finally, the original frequency vector satisfies the decomposition
$\omega(t)-\omega_{\rm nom}\mathbf{1}_n = \tilde{\omega}(t)+\mathbf{1}_n e_\xi(t)$. Therefore,
$\|\omega(t)-\omega_{\rm nom}\mathbf{1}_n\| \leq
\kappa e^{-\lambda t}\|\tilde{x}(0)\|
+ \gamma
\sup_{0\leq s\leq t}\|\Pi d(s)\| +
\sqrt{n} e^{-t/\tau_P}\|e_\xi(0)\| + \sqrt{n} \tau_P
\sup_{0\leq s\leq t}\|d_\xi(s)\|$.
Equivalently, if $ \sup_{t\geq 0}\|\Pi d(t)\|\leq \Delta_d,
\sup_{t\geq 0}|d_\xi(t)| \leq \Delta_\xi$,
then $\|\omega(t)-\omega_{\rm nom}\mathbf{1}_n\|
\leq
\kappa e^{-\lambda t}\|\tilde{x}(0)\| 
+
\sqrt{n} e^{-t/\tau_P}|e_\xi(0)| +
\gamma\Delta_d
+
\sqrt{n}\tau_P\Delta_\xi$. Thus, the frequency deviation from nominal synchronization is input-to-state practically stable~\cite{jiang1994small} with respect to the projected input $\Pi d$ and the weighted average-frequency mismatch $d_\xi$.
\end{theorem}

\begin{proof}
We prove the result in three steps. First, we show that the projected state remains in the disagreement subspace. Second, we apply the ISS estimate for the projected dynamics on this subspace from Theorem~\ref{thm:isps_projected_frequency}. Third, we combine this estimate with the scalar weighted-average frequency dynamics.

For any initial condition $(\delta(0),\omega(0))$, projected variables $(\tilde{\delta}(0), \tilde{\omega}(0))$ satisfy $\xi^\top\tilde{\delta}(0)
= \xi^\top\Pi\delta(0)
=0, \xi^\top\tilde{\omega}(0)=
\xi^\top\Pi\omega(0)
=0.$
Thus,
$\tilde{x}(0)=
\begin{bmatrix}
\tilde{\delta}(0) \ \
\tilde{\omega}(0)
\end{bmatrix}
\in\mathcal S.$

As shown in Theorem~\ref{thm:isps_projected_frequency}, $\mathcal{S}$ is invariant under the projected dynamics~\eqref{eq:projected_compact}.
Therefore, if $\tilde{x}(0)\in\mathcal S$, then $\tilde{x}(t)\in\mathcal S$ for all $t\geq 0$. Hence, the projected trajectory evolves entirely on the disagreement subspace, where the restriction of $\mathbf A_\Pi$ in~\eqref{eq:projected_compact} is Hurwitz by Lemma~\ref{lem:A_pi_hurwitz}.

Consequently, the ISS estimate on $\mathcal{S}$ applies to the projected system. Thus, from Theorem~\ref{thm:isps_projected_frequency}, there exist constants $\kappa\geq 1$, $\lambda>0$, and $\gamma>0$ such that
\[
\|\tilde{x}(t)\|
\leq
\kappa e^{-\lambda t}\|\tilde{x}(0)\|
+
\gamma\sup_{0\leq s\leq t}\|\Pi d(s)\|,
\quad t\geq 0,
\]
where, $\Pi d$ is the projected input. Since $\tilde{\omega}$ is a component of $\tilde{x}$, it follows that
\[
\|\tilde{\omega}(t)\|
\leq
\|\tilde{x}(t)\|
\leq
\kappa e^{-\lambda t}\|\tilde{x}(0)\|
+
\gamma\sup_{0\leq s\leq t}\|\Pi d(s)\|.
\]

It remains to analyze the weighted-average frequency component. Define $\omega_\xi(t):=\xi^\top\omega(t),
\
e_\xi(t):=\omega_\xi(t)-\omega_{\rm nom}$. Using the original frequency dynamics
$\dot{\omega} = \mathbf{G}_\omega\delta
-\frac{1}{\tau_P}\omega
+ d(t)$, where, $d(t)$ is as in \eqref{eq:projected_compact}. Multiplying by $\xi^\top$, we obtain
\begin{align*}
\dot{\omega}_\xi
& \textstyle = \xi^\top\dot{\omega} =
\xi^\top\mathbf{G}_\omega\delta
-\frac{1}{\tau_P}\xi^\top\omega
+
\xi^\top d(t).
\end{align*}
Since $\xi^\top\mathbf{G}_\omega=0$, this reduces to
$ \dot{\omega}_\xi = 
-\frac{1}{\tau_P}\omega_\xi
+
\xi^\top d(t).$ 
Subtracting $\omega_{\rm nom}$ from both sides gives
\[
\dot e_\xi = \textstyle 
-\frac{1}{\tau_P}e_\xi
+
\left(
\xi^\top d(t)-\frac{1}{\tau_P}\omega_{\rm nom}
\right).
\]
Define, $ d_\xi(t):=
\xi^\top d(t)-\frac{1}{\tau_P}\omega_{\rm nom}$. Then $\dot e_\xi = -\frac{1}{\tau_P}e_\xi+d_\xi(t)$.
By the variation-of-constants formula,
\[ \textstyle 
e_\xi(t) = 
e^{-t/\tau_P}e_\xi(0)
+
\int_0^t e^{-(t-s)/\tau_P}d_\xi(s)ds.
\]
Taking absolute values yields
\begin{align*}
|e_\xi(t)|
&\leq \textstyle 
e^{-t/\tau_P}|e_\xi(0)|
+
\int_0^t e^{-(t-s)/\tau_P}|d_\xi(s)|ds \\
& \textstyle \leq
e^{-t/\tau_P}|e_\xi(0)|
+
\sup_{0\leq s\leq t}|d_\xi(s)|
\int_0^t e^{-(t-s)/\tau_P}ds \\
& \textstyle \leq
e^{-t/\tau_P}|e_\xi(0)|
+
\tau_P\sup_{0\leq s\leq t}|d_\xi(s)|.
\end{align*}
Finally, decompose the original frequency vector as
$ \omega
= \Pi\omega+\mathbf{1}_n\xi^\top\omega = 
\tilde{\omega}+\mathbf{1}_n\omega_\xi.$
Therefore,
\[
\omega-\omega_{\rm nom}\mathbf{1}_n = 
\tilde{\omega}
+
\mathbf{1}_n(\omega_\xi-\omega_{\rm nom}) = \tilde{\omega}
+
\mathbf{1}_n e_\xi.
\]
Using the triangle inequality,
\begin{align*}
\|\omega(t)-\omega_{\rm nom}\mathbf{1}_n\|
&\leq
\|\tilde{\omega}(t)\|
+
\|\mathbf{1}_n e_\xi(t)\| \\
& =
\|\tilde{\omega}(t)\|
+
\sqrt{n}\|e_\xi(t)\|.
\end{align*}
Substituting the bounds on $\tilde{\omega}(t)$ and $e_\xi(t)$ gives
\begin{align*}
\|\omega(t)-\omega_{\rm nom}\mathbf{1}_n\|
&\leq
\kappa e^{-\lambda t}\|\tilde{x}(0)\|
+
\gamma
\sup_{0\leq s\leq t}\|\Pi d(s)\|\\
&
\hspace{-0.2in} +
\sqrt{n}e^{-t/\tau_P}|e_\xi(0)|
+
\sqrt{n}\tau_P
\sup_{0\leq s\leq t}|d_\xi(s)|.
\end{align*}
This proves the stated bound on the frequency. If, in addition,
$ \sup_{t\geq 0}\|\Pi d(t)\|\leq \Delta_d, \sup_{t\geq 0}|d_\xi(t)|\leq \Delta_\xi$,
then the preceding estimate immediately implies
\begin{align*}
\|\omega(t)-\omega_{\rm nom}\mathbf{1}_n\|
&\leq
\kappa e^{-\lambda t}\|\tilde{x}(0)\|
+
\sqrt{n}e^{-t/\tau_P}|e_\xi(0)| \\
& \hspace{0.2in} +
\gamma\Delta_d
+
\sqrt{n}\tau_P\Delta_\xi.
\end{align*}
Thus, the frequency deviation from nominal synchronization is input-to-state practically stable with respect to the disagreement input $\Pi d$ and the weighted average-frequency mismatch $d_\xi$.
\end{proof}

\section{Conclusion}
This paper presented a mathematical stability analysis of a sampled-data optimization-based secondary controller for networks of inverter-interfaced DERs. The analysis treated the controller model as given and studied the nonlinear closed-loop dynamics induced by sampled measurements, constrained optimization updates, and interpolation-based actuation. For the GFM-DER voltage loop, we established large-signal boundedness of the voltage and filtered reactive-power dynamics within a certified operating region. We also characterized positive steady-state operating points and showed how the optimization-induced consensus condition connects voltage regulation with equal per-unitized reactive power sharing. For the phase-frequency dynamics, we used a projected representation to remove the marginal absolute-angle mode and established input-to-state stability with respect to active-power mismatch. These results provide closed-loop guarantees for optimization-based secondary control beyond small-signal or purely continuous-time analyses. Future work will focus on relaxing some of the technical assumptions used in the analysis, including identical active-power filter time constants and exact tracking by the inner GFL-DER control loop. 
\appendix
\subsection{Supporting Lemmas}
\begin{lemma}\label{lem:psd_proof}
 Let $\mathcal{B} \in \mathbb{R}^{n\times n}$ be a square matrix with $[\mathcal{B}]_{ii} := -B_{i i} = -(B^{\mathrm{sh}}_{i} + \sum_{\kk \in N_i} B_{ik}), [\mathcal{B}]_{i\kk} := B_{i\kk}, \kk \neq i$, then the matrix $\Sigma \mathcal{B} + \mathcal{B}^\top \Sigma$ with $\Sigma := \text{diag}((\tilde{\beta}_i r_{V_i}V_i)/2\tau^2_{Q_i})$ is positive semi-definite.
\end{lemma}

\begin{proof} For any $x \in \mathbb{R}^{n}$ we have
$x^\top (\Sigma \mathcal{B} + \mathcal{B}^\top \Sigma)x = 2x^\top \Sigma \mathcal{B} x$,
since $x^\top \Sigma \mathcal{B} x$ is a scalar and equal to its transpose. Since, $V_i > 0$, for all $i$, thus, $[\Sigma]_{i i} > 0$ for all $i$. Let $y = \Sigma^{1/2} x$ (well-defined since $\Sigma$ is a positive diagonal matrix). Then $x = \Sigma^{-1/2} y$ and $
2 x^\top \Sigma \mathcal{B} x 
= 2 y^\top ( \Sigma^{-1/2} \mathcal{B} \Sigma^{-1/2} ) y.$
Define $\nu := \Sigma^{-1/2} \mathcal{B} \Sigma^{-1/2}$. Note that $\nu = \sigma^\top \mathcal{B} \sigma$ with $\sigma = \Sigma^{-1/2}$, so $\nu$ is congruent to $\mathcal{B}$, and $\nu \succeq 0$ if and only if $\mathcal{B} \succeq 0$. Therefore, $\Sigma \mathcal{B} + \mathcal{B}^\top \Sigma \succeq 0 \iff \mathcal{B} \succeq 0$, which is indeed the case due to the Gershgorin Circle Theorem.
\end{proof}

\begin{lemma}\label{lem:continuous_sol}
Solution $x_\s$ of~\eqref{eq:gfm_optimization} is continuous in $(V,Q(V))$.  
\end{lemma}
\begin{proof}
Due to the consensus constraints, every feasible point has the form
$x_s=c\mathbf 1_{n}$. Hence, the optimization reduces to the
scalar problem $ \min_{c\in\mathcal I(V,Q)}
\|c\mathbf 1_{n}-\alpha(V,Q)\|_2^2$, 
where $\mathcal I(V,Q) := \textstyle 
\bigcap_{i=1}^{n}
[m_i(V,Q)-V_\Delta, m_i(V,Q)+V_\Delta]$ and $ m_i(V,Q)
:= (1+\beta_{Q_i})r_{V_i}Q_i^{\mathrm{avg}}(t_\s)
- \beta_{V_i}(\Vnom - V_i (t_\s)).$
Writing $\bar\alpha(V,Q):=
\frac{1}{n}\mathbf 1^\top\alpha(V,Q),$
we have $\|c\mathbf 1-\alpha(V,Q)\|_2^2
=
n(c-\bar\alpha(V,Q))^2
+
\|\alpha(V,Q)-\bar\alpha(V,Q)\mathbf 1\|_2^2.$ Therefore, the optimizer is $c_\s(V,Q)=\Pi_{\mathcal I(V,Q)}(\bar\alpha(V,Q))$.
Let $\ell(V,Q):=\max_i\{m_i(V,Q)-V_\Delta\},
r(V,Q):=\min_i\{m_i(V,Q)+V_\Delta\}.$
Then $\mathcal I(V,Q)=[\ell(V,Q),r(V,Q)]$. Since each $m_i(V,Q)$ is
continuous, $\ell$ and $r$ are continuous. Moreover, $\bar\alpha$ is
continuous because $\alpha$ is continuous. Thus
$c_\s(V,Q)
=
\min\{\max\{\bar\alpha(V,Q),\ell(V,Q)\},r(V,Q)\}$ 
is continuous. Hence, the solution
$ x_\s(V,Q)=c_\s(V,Q)\mathbf 1_{n}$ 
is continuous on the domain where $\mathcal I(V,Q)$ is nonempty.
\end{proof}

\begin{lemma}\label{lem:A_pi_hurwitz}
Suppose that the following conditions hold: i) $\tau_{P_i}=\tau_P>0 , r_{\omega_i}>0, V_i>0$ for all $i\in \{1,\dots,n\}$, ii) the GFM-DER interaction graph is connected. Then the restriction of $\mathbf{A}_\Pi$ to the disagreement subspace $ \mathcal{S}
:=
\left\{
\tilde{x}=
\begin{bmatrix}
\tilde{\delta}\\
\tilde{\omega}
\end{bmatrix}
\in\mathbb{R}^{2n}
|
\xi^\top\tilde{\delta}=0,
\xi^\top\tilde{\omega}=0
\right\}$, where $\xi$ is a normalized left null vector of $\mathbf{G}_\omega$, is Hurwitz.
\end{lemma}
\begin{proof}
Define the positive diagonal matrix $R_\omega
:=
\operatorname{diag}(r_{\omega_1},\dots,r_{\omega_{n}})$. Next, define the symmetric weighted Laplacian $L_V$ associated with the GFM-DER interaction graph by $[L_V]_{i\jj}
= B_{i\jj}V_i V_\jj,
\ i\neq \jj, \ \mbox{and} \
[L_V]_{ii}
= -\sum_{\jj\neq i}B_{i\jj}V_i V_\jj.$
Since $B_{ij}=B_{ji}<0$ on every edge and $V_i>0$, $L_V$ is a symmetric positive semidefinite weighted Laplacian. Since the GFM-DER interaction graph is connected, $
\operatorname{ker}(L_V)=\operatorname{span}\{\mathbf{1}_{n}\}.$ Using the definition of $\mathbf{G}_\omega$, we can write $\mathbf{G}_\omega = -\frac{1}{\tau_P}R_\omega L_V$. Because $R_\omega\succ 0$, the matrix $R_\omega L_V$ is similar to a symmetric positive semidefinite matrix. Indeed,
$R_\omega L_V
=
R_\omega^{1/2}
\left(
R_\omega^{1/2}L_VR_\omega^{1/2}
\right)
R_\omega^{-1/2}$. Thus, $R_\omega L_V$ has real nonnegative eigenvalues. Consequently, $\mathbf{G}_\omega$ has real nonpositive eigenvalues. Since the graph is connected, $\mathbf{G}_\omega$ has one zero eigenvalue corresponding to the uniform-angle direction and all remaining eigenvalues are strictly negative. That is,
$ \lambda_1(\mathbf{G}_\omega)=0,
\ \lambda_k(\mathbf{G}_\omega)<0,
\ k=2,\dots,n.
$ The zero eigenvalue corresponds to the direction $\mathbf{1}_{n}$, which is removed by the projection. Hence, on the disagreement subspace $\mathcal{S}$, only the modes associated with $\lambda_k(\mathbf{G}_\omega)<0$ remain. 

To show that this implies the Hurwitz property of $\mathbf{A}_\Pi$ on the disagreement subspace, we now lift the modal properties of $\mathbf{G}_\omega$ to the second-order phase-frequency dynamics. Since $\mathbf{G}_\omega$ is similar to a symmetric matrix, it is diagonalizable and has real eigenvalues. Let $v_k$ be an eigenvector of $\mathbf{G}_\omega$ associated with a disagreement eigenvalue $\lambda_k<0$, so that
$ \mathbf{G}_\omega v_k=\lambda_k v_k.$ 
Consider an eigenpair 
$\left(s,
    \begin{bmatrix}
        \phi\\
        \psi
    \end{bmatrix}
    \right)$
of $\mathbf{A}_\Pi$ on the disagreement subspace. Then
$\begin{bmatrix}
        0_{n} & \mathbb{I}_{n}\\
        \mathbf{G}_\omega & -\frac{1}{\tau_P}\mathbb{I}_{n}
    \end{bmatrix}
    \begin{bmatrix}
        \phi\\
        \psi
    \end{bmatrix}
    =
    s
    \begin{bmatrix}
        \phi\\
        \psi
    \end{bmatrix}.
$    
The first block row gives $\psi=s\phi$. Substituting this relation into the second block row gives $    \mathbf{G}_\omega\phi-\frac{1}{\tau_P}\psi=s\psi.$ 
Using $\psi=s\phi$, we obtain $\mathbf{G}_\omega\phi-\frac{s}{\tau_P}\phi=s^2\phi$,
or equivalently, $\mathbf{G}_\omega\phi
    = \left(s^2+\frac{s}{\tau_P}\right)\phi$. 
Thus, for each eigenvalue $\lambda_k$ of $\mathbf{G}_\omega$, the corresponding eigenvalues $s$ of $\mathbf{A}_\Pi$ satisfy $    s^2+\frac{1}{\tau_P}s-\lambda_k=0.$
On the disagreement subspace, $\lambda_k<0$. Hence, $\frac{1}{\tau_P}>0,
    \ -\lambda_k>0.$
Therefore, by the second-order Routh-Hurwitz criterion, the polynomial $s^2+\frac{1}{\tau_P}s-\lambda_k$. 
has both roots in the open left-half complex plane. Consequently, every eigenvalue of $\mathbf{A}_\Pi$ associated with a disagreement mode has a strictly negative real part. The only eigenvalue of $\mathbf{G}_\omega$ that is not strictly negative is $\lambda_1=0$, which corresponds to the uniform-angle direction $\mathbf{1}_{n}$. For this mode, the characteristic equation becomes
$s^2+\frac{1}{\tau_P}s=0,$
whose roots are $s=0, s=-\frac{1}{\tau_P}.$
The zero root corresponds to the absolute-angle mode, which is removed by projection onto $\mathcal{S}$. Therefore, no marginal mode remains on $\mathcal{S}$, and all eigenvalues of $\mathbf{A}_\Pi$ restricted to $\mathcal{S}$ have strictly negative real parts. Hence, $\mathbf{A}_\Pi$ restricted to $\mathcal{S}$ is Hurwitz.
\end{proof}

\bibliography{references}

\end{document}